\documentclass[runningheads,a4paper]{llncs}

\usepackage{listings}
\usepackage{etex}

\usepackage{tikz}
\usetikzlibrary{calc,arrows,shapes,chains,matrix,positioning,scopes,decorations.pathmorphing,shadows,fit,mindmap}
\definecolor{rosso}{rgb}{.99,.5,.5}
\definecolor{dgreen}{rgb}{0,.7,0}
\usetikzlibrary{matrix}
\usetikzlibrary{arrows}
\usetikzlibrary{automata}

\usepackage{pgf}
\usepackage{amsmath,amssymb}
\usepackage{arrows}

\usepackage{breakurl}
\newcommand{\url}[1]{\textsf{#1}}
\usepackage{latexsym}
\usepackage{textcomp}
\usepackage{subfigure}
\usepackage{graphicx}

\usepackage[arrow,matrix,frame,curve,cmtip]{xy}

\newcommand{\checkDtmc}{\textsf{Check}}
\newcommand{\checkPath}{\textsf{CheckPath}}

\newcommand{\relop}{\bowtie}

\def\sosrule#1#2{\frac{\,\,\,\,#1\,\,\,\,}{\,\,\,\,#2\,\,\,\,}}

\def\trans#1{\stackrel{#1}{\rightarrowtail}}

\usepackage[algoruled,linesnumbered,lined,english]{algorithm2e}

\usepackage{color}
\usepackage{colortbl}
\usepackage{mathrsfs}
\definecolor{apricot}{RGB}{249,243,221}
\def\RED#1{\textcolor{red}{#1}}

\def\endstat{\hfill\bullet}

\newcommand{\nxt}{{\mathsf{next}}}

\newcommand{\lab}{{\mathsf{lab\_eval}}}

\newcommand{\AND}{\mbox{$\;\;\wedge\;\;$}}

\newcommand{\All}{\mbox{$\forall\,$}}
\newcommand{\Ex}{\mbox{$\exists\,$}}

\newcommand{\true}{\mbox{tt}}

\def\Prob#1#2#3{{\calP}_{#1 #2}(#3)}

\def\rU#1{\, {\calU}^{#1} \,}

\def\FlyFast#1{{\sf FlyFast#1}}

\newcommand{\SET}[1]{\{#1\}}
\newcommand{\aN}{^{(N)}}

\def\calA{{\cal A}}

\def\calD{{\cal D}}

\def\calH{{\cal H}}

\def\calM{{\cal M}}

\def\calP{{\cal P}}

\def\calS{{\cal S}}

\def\calU{{\cal U}}

\def\calX{{\cal X}}

\renewcommand{\iff}{\mbox{ iff }}

\newcommand{\at}{\mbox{\sl L}}
\newcommand{\Paths}{\mbox{\sl Paths}}

\renewcommand{\Pr}{{\mathbb{P}}}

\def\calS{{\cal S}}

\def\act#1{{\calA}_{#1}} 	
\def\bop{\langle \mbox{bop}\rangle}	
\def\dfas{:=}			
\def\dirac{\boldsymbol{\delta}}
\def\expsem{\mathscr{E}}  
\def\frc{\textsf{frc}\,}
\def\hdtm{\mathbf{H}}	

\def\map{\calH}
\def\nats{\mathbb{N}}	
\def\lnxt{\calX \,}			
\def\otm{\mathbf{K}}		
\def\os{\Delta}			
\def\relop{\langle \mbox{relop}\rangle}	
\def\semint#1#2{#1[\![#2]\!]}
\def\sep{|}				
\def\sc#1{\calS_{#1}}		
\def\slf{\ell}			
\def\sls{\mathscr{P}}			
\def\stm{\mathbf{S}}		
\def\sdtm{\mathbf{P}}	
\def\reals{\mathbb{R}}	
\def\uop{\langle \mbox{uop}\rangle}	
\def\us#1{\calU^#1}
\def\vct#1{{\mathbf #1}}

\def\deriv{\noindent\hspace*{.20in}\vspace{0.1in}}

\def\hint#1#2{\newline#1\hspace*{.25in}\{$#2$\}\vspace{0.1in}\newline\hspace*{.20in}\vspace{0.1in}}

\newtheorem{nfa}{\RED{Note for Authors}}[section]

\newif\ifTR
\TRtrue
\begin{document}

\hyphenation{System-of-In-de-pend-ent}

\mainmatter




\ifTR
\title{On-the-fly Fast Mean-Field Model-Checking\\Extended version
}
\titlerunning{On-the-fly Fast Mean-Field  Model-Checking
\thanks{This research has been partially funded by the EU 
projects ASCENS (nr. 257414) and QUANTICOL (nr. 600708), 
and the IT MIUR project CINA.}}
\else
\title{On-the-fly Fast Mean-Field Model-Checking
\thanks{This research has been partially funded by the EU 
projects ASCENS (nr. 257414) and QUANTICOL (nr. 600708), 
and the IT MIUR project CINA.}}
\titlerunning{On-the-fly Fast Mean-Field  Model-Checking}
\fi

\author{Diego~Latella\inst{1} \and Michele Loreti\inst{2}  \and  Mieke~Massink\inst{1}}
\institute{Istituto di Scienza e Tecnologie dell'Informazione \lq A.~Faedo\rq, CNR, Italy \and
Universit\`a di Firenze, Italy
}

\authorrunning{Latella et al.}

\maketitle

\noindent

\begin{abstract}
A novel, scalable, on-the-fly model-checking procedure is presented to verify bounded PCTL properties of selected individuals in the context of very large systems of independent interacting objects. The proposed procedure combines on-the-fly model checking techniques with deterministic mean-field approximation in discrete time. The asymptotic correctness of the procedure is shown and some results of the application of a prototype implementation of the \FlyFast{} model-checker are presented.
\end{abstract}



\renewcommand{\RED}[1]{#1}

\section{Introduction}
\label{Introduction}


Model checking has been widely recognised as a powerful approach to the automatic verification of concurrent and distributed systems. It consists of an efficient procedure that, given an abstract model $\calM$ of the system, decides whether $\calM$ satisfies a logical formula $\Phi$, typically drawn from a temporal logic. Despite the success of model-checking 
procedures, their scalability have always been a concern due to the potential combinatorial explosion of the state space that needs to be searched. 
%

The main contribution of this paper is a novel model-checking procedure, based on an original combination of local, on-the-fly model-checking techniques and mean field approximation in discrete time~\cite{BMM07}. The procedure can be used to verify bounded PCTL~\cite{Hansson1994} properties of selected individuals in the context of systems consisting of a large number of similar but independent interacting objects. It is scalable in the sense that it is insensitive to the size of the population the system consists of. 
The asymptotic correctness of the model-checking procedure is proven and a prototype implementation of the model-checker, \FlyFast{,} is applied to a bench-mark example from computer epidemics that was also studied extensively in~\cite{Bo+13}, to which we refer for a detailed discussion. To the best of our knowledge, this is the first implementation of an {\em on-the-fly mean field} model-checker for {\em discrete time, probabilistic, time-synchronous} models.

Following the approach in~\cite{BMM07} we consider a model for interacting objects, where the evolution of each object is given by a finite state discrete time Markov chain. The transition matrix of each object may depend on the distribution of states of all objects in the system. Each object can be in one of its local states at any point in time and all objects proceed in discrete time and in a clock-synchronous fashion. When the number of objects is large, the overall behaviour of the system in terms of its `occupancy measure', i.e. the fraction of objects that are in a particular local state at a particular time, can be approximated by the (deterministic) solution of a {\em difference equation} which is called the `mean field'\footnote{The term `mean field' has its origin in statistical physics and is sometimes used with slightly different meaning in the literature. Here we intend the meaning as defined in~\cite{BMM07}.}. This convergence result has been extended in~\cite{BMM07} to obtain a `fast' way to stochastically simulate the evolution of a selected, limited number of specific objects in the context of the overall behaviour of the population. 

We show that the deterministic iterative procedure of~\cite{BMM07}, to compute the average overall behaviour of the system and that
of individual objects in the context of the overall system, combines well with an {\em on-the-fly} probabilistic model-checking procedure for the verification of bounded PCTL formulas addressing selected objects of interest\footnote{Note that the transition probabilities of these selected objects at time $t$ may depend on the occupancy measure of the system at $t$ and therefore also the truth-values of the formulas may vary with time.}. An on-the-fly recursive approach also provides a natural way to address nested path formulae and time-varying truth values of such formulae. The algorithm presented in this paper is parametric w.r.t. the semantic interpretation of the language. In particular we present two different interpretations; one based on the standard, {\em exact probabilistic semantics} of a simple probabilistic population description language, and the other one on the {\em mean-field approximation in discrete time} of such a semantics. The latter is the main contribution of the current paper. The considered PCTL formulae can be extended along the lines proposed in~\cite{Ko+12,Ko+13} with properties that address the overall status of the system. We show a simple instance of that.

The models we consider are also known as SIO-models (System of Independent Objects)~\cite{Bo+13}. These are time-synchronous models in which each object performs a probabilistic step in each discrete time unit, possibly looping to the same state. This is a class of models that is frequently encountered in various research disciplines ranging from telecommunication to computational biology. The objects interact in an indirect way via the global state of the overall system.  

\section{Related Work}

Traditionally, model checking approaches are divided into two broad categories: {\em global} approaches that determine the set of {\em all} states in $\calM$ that satisfy $\Phi$, and {\em local} approaches that, given a state $s$ in $\calM$, determine whether $s$ satisfies $\Phi$~\cite{Courcoubetis1992,BCG95}. 

Global symbolic model checking algorithms are popular because of their computational efficiency and can be found in many model checkers, both in a qualitative (see e.g. ~\cite{Clarke1986})  and in a stochastic setting (see e.g. ~\cite{Ba+03,KNP04}). The set of states that satisfy a formula is constructed recursively in a {\em bottom-up} fashion following the syntactic structure of the formula. Depending on the particular formula to verify, usually the underlying model can be reduced to fewer states before the algorithm is applied. Moreover, as is shown e.g. in~\cite{Ba+03} for stochastic model checking, the model checking algorithm can be reduced to combinations of existing well-known and optimised algorithms for CTMCs such as transient analysis.

Local model checking algorithms have been proposed to mitigate the state space explosion problem using a so called `on-the-fly' approach (see e.g.~\cite{Courcoubetis1992,BCG95,Hol04,GnM11}). 
On-the-fly algorithms are following a {\em top-down} approach that does not require global knowledge of the complete state space. For each state that is encountered, starting from a given state, the outgoing transitions are followed to adjacent states, constructing step by step local knowledge of the state space until it is possible to decide whether the given state satisfies the formula. 
For qualitative model checking, local model-checking algorithms have been shown to have the same worst-case complexity as the best existing global procedures for the above mentioned logics. However, in practice, they have better performance  when only a subset of the system states need to be analysed to determine whether a system satisfies a formula. Furthermore, local model-checking may still provide some results in case of systems with a very large or even infinite state space where global model checking approaches would be impossible to use. In the context of stochastic model checking several on-the-fly approaches have been proposed, among which~\cite{DIMTV04} and~\cite{Ha+09}. The former is a probabilistic model checker for bounded PCTL formulas. The latter uses an on-the-fly approach to detect a maximal relevant search depth in an infinite state space and then uses a {\em global} model-checking approach to verify bounded CSL~\cite{Az+00,Ba+03} formulas in a continuous time setting on the selected subset of states. 
An on-the-fly approach by itself however, does not solve the challenging scalability problems that arise in truly large parallel systems, such as collective adaptive systems, e.g., gossip protocols~\cite{CBR09}, self-organised collective decision making~\cite{MonFerSchPinBirDor11}, computer epidemics~\cite{BGH08} and foreseen smart urban transportation systems and decentralised control strategies for smart grids.

To address this type of scalability challenges in probabilistic model-checking, recently, several approaches have been proposed. In~\cite{He+04,Gu+06} approximate probabilistic model-checking is introduced. This is a form of statistical model-checking that consists in the generation of random executions of an {\em a priori} established maximal length. On each execution the property of interest is checked and statistics are performed over the outcomes. The number of executions required for a reliable result depends on the maximal error-margin of interest. The approach relies on the analysis of individual execution traces rather than a full state space exploration and is therefore memory-efficient. However, the number of execution traces that may be required to reach a desired accuracy may be large and therefore time-consuming. The approach works for general models, i.e., not necessarily populations of similar objects, but is not independent of the number of objects involved.

To analyse properties of large scale mobile communication networks mean field approximations in discrete time have also been used e.g., in Bakshi et al.~\cite{Ba+10}. In that work an automatised method is proposed and applied to the analysis of dynamic gossip networks. A general convergence result to a deterministic difference equation is used, similar to that in~\cite{BMM07}, but not its extension to analyse individual behaviour in the context of a large population, nor its exploitation in model-checking algorithms.

In Chaintreau et al.~\cite{CBR09}, mean field convergence in {\em continuous time} is used to analyse the distribution of the age of information that objects possess when using a mix of gossip and broadcast for information distribution in situations where objects are not homogeneously distributed in space. An overview of mean field interaction models for computer and communication systems by Bena\"{i}m et al. can be found in~\cite{BeB08}.

Preliminary ideas on the exploitation of mean field convergence in {\em continuous time} for model-checking mean field models, and in particular for an extension of the logic CSL, were informally sketched in a presentation at QAPL 2012~\cite{Ko+12}, but no model-checking algorithms were presented. Follow-up work on the above mentioned approach can be found in~\cite{Ko+13} which relies on earlier results on fluid model checking by Bortolussi and Hillston~\cite{BoH12b}. In the latter a {\em global CSL} model-checking procedure is proposed for the verification of properties of a selection of individuals in a population.  This work is perhaps the most closely related to our work, however their procedure exploits mean field convergence and fast 
simulation~\cite{DaN08,GaG10} in a {\em continuous} time setting
rather than in a discrete time setting and is based on an {\em interleaving} model of computation, rather than
a clock-synchronous one; furthermore, a {\em global } model-checking approach,
rather than an on-the-fly approach,  is followed. The modelling language used in~\cite{BoH12b} is PEPA.
%
%
Earlier work by Stefanek et al.~\cite{SHB10} on the use of mean field convergence in {\em continuous time} for grouped PEPA has investigated the quality of the convergence results when the related differential equations are derived directly from the process algebraic model. Potential issues with accuracy were found concerning the parallel composition operator of PEPA that involves a (non-linear) minimum function applied to rates originating from synchronising populations. This could, in some circumstances, give rise to  inaccuracies in the approximation. It is however possible to detect such situations.

\section{Time bounded PCTL and On-the-fly Model-Checking}
\label{PCTLaMC}

In this section we  recall the definition of the {\em time bounded} fragment of PCTL\footnote{For notational simplicity we  call the fragment PCTL as well.} and we present an on-the-fly model-checking algorithm. The algorithm is
parametric in the sense that it  can be used for different languages and semantic interpretations. In this paper we use two instantiations of the algorithm; one is
on a DTMC semantics of a simple language of object populations (Sect.~\ref{ModellingLanguage}) and the other is on a mean-field approximation semantics of the same language, for ``fast model-checking'' (Sect.\ref{MeanFieldAprx}).
For the sake of readability, we present only a schema of the algorithm 
 for time {\em bounded} PCTL, that is the same as that proposed in~\cite{DIMTV04}. 
 The interested reader is referred to~\cite{La+13b} where 
 a novel algorithm is defined and implemented for {\em the full} logic.


\subsection{Time bounded PCTL}
\label{Logic}

Given a  set  $\sls$ of atomic propositions, the syntax of PCTL is defined 
below, where $a\in \sls$, $k\geq 0$ and $\bowtie\, \in \SET{\geq,>,\leq,<}$:
\[
\Phi  ::=  a \mid \neg \, \Phi \mid  \Phi \, \vee \, \Phi \mid \Prob{\bowtie}{p}{\varphi} \qquad
\mbox{ where }
\varphi ::=  \lnxt \Phi \mid  \Phi \rU{\le k} \Phi .
\]
%
PCTL formulae are interpreted over {\em state labelled}  DTMCs.
A state labelled DTMC is a pair $\langle \calM, \slf  \rangle$ where $\calM$ is
a DTMC with state set $\calS$ and $\slf : \calS \rightarrow 2^\sls$ associates each state
with a set of atomic propositions; for each state $s \in \calS$, $\ell(s)$ is the
set of atomic propositions true in $s$. In the following, we assume $\sdtm$ 
be the one step probability matrix for $\calM$; we abbreviate
$\langle \calM, \slf  \rangle$ with $\calM$, when no confusion can arise.
A path $\sigma$ over  $\calM$ is a non-empty sequence of states
$s_0, s_1, \cdots$ where  $\sdtm_{s_i,s_{i+1}} >0$
for all $i\ge 0$. 
We let $\Paths_{\calM}(s)$ denote the set of all infinite paths over $\calM$ starting from state $s$. 
By $\sigma[i]$ we denote the $i$-th  element $s_i$ of path 
$\sigma$. 
Finally, in the sequel we will consider DTMCs equipped with an initial state $s_0$,
i.e. the probability mass is initially all in $s_0$. For any such a DTMC $\calM$, and for all $t \in \nats$ 
we let the set $\at_{\calM}(t) = \{\sigma[t] \mid \sigma \in \Paths_{\calM}(s_0)\}$.

We define the satisfaction relation on $\calM$ and the logic in Table~\ref{D:DPCTLSAT}.
\begin{table}[tbp]
\begin{center}
\fbox{
\begin{minipage}{4.5in}
\footnotesize
$
\begin{array}{lcl}
s \models_{\calM} a & \iff & a \in \ell(s) \\[1ex]
s \models_{\calM} \neg \Phi & \iff & \mbox{not } s \models_{\calM} \Phi \\[1ex]
s \models_{\calM} \Phi_1 \, \vee \, \Phi_2 & \iff & s \models_{\calM} \Phi_1 
\mbox{ or } s \models_{\calM} \Phi_2\\[1ex]
s \models_{\calM} \Prob{\bowtie}{p}{\varphi} & \iff &
   \Pr \{ \sigma \in \Paths_{\calM}(s) \mid \sigma \models_{\calM} \varphi \} \bowtie p\\[0.20cm]
\sigma \models_{\calM} \, \lnxt \, \Phi & \mbox{iff  } & \sigma[1] \models_{\calM} \Phi \\
\sigma  \models_{\calM} \Phi_1 \rU{\le k}{} \Phi_2 & 
\mbox{iff  } & 
\Ex 0 \le h \le k \;s.t.\; \sigma[h] \models_{\calM} \Phi_2 \AND 
\All 0\le i <  h \;.\;  \sigma[i] \models_{\calM} \Phi_1\\
\end{array}
$
\normalsize
\end{minipage}
}
\end{center}
\caption{Satisfaction relation for Time Bounded PCTL.\label{D:DPCTLSAT}}
\end{table}

\subsection{On-the-fly PCTL Model-Checking Algorithm}


In this section we introduce a local on-the-fly model-checking algorithm for time-bounded PCTL formulae. The basic 
idea of an on-the-fly algorithm is simple: while the state space is generated in a stepwise fashion from a term $s$ of the language, the algorithm considers only the relevant prefixes of the paths while they are generated. 
For each of them it updates the information about the satisfaction of the formula 
that is checked. In this way, only that part of the state space is generated that can provide information on the satisfaction 
of the formula and irrelevant parts are not taken into consideration. 

In the case of probabilistic process population languages, for large populations,
a mean-field approximated semantics can be defined. In Sect.~\ref{MeanFieldAprx}
we show how a drastic reduction of the state space can be obtained, by using the same algorithm on such semantic models. We call such a combined use
of on-the-fly model-checking and mean-field semantics ``Fast model-checking''
after ``Fast simulation'', introduced in~\cite{BMM07}.

The algorithm abstracts from any specific language and different semantic interpretations of a language. We only assume an abstract interpreter function that, given a generic process term, returns a probability distribution over the set  of terms. 
Below, we let \textsf{proc} be the (generic) type of \emph{probabilistic process terms}
while we let \textsf{formula} and \textsf{path\_formula} be the types of \emph{state-} and \emph{path-}
PCTL formulae. Finally, we use  \textsf{lab} to denote the type of \emph{atomic propositions}.

The abstract interpreter can be modelled by means of two functions: 
$\nxt$ and $\lab$.
Function $\nxt$ associates a list of pairs $(\mathsf{proc},\mathsf{float})$ to each element of type \textsf{proc}. The list of pairs gives the terms, i.e. states, that can be reached in one step from the given state and their one-step transition probability. We require that for each $s$ of type $\mathsf{proc}$ it  holds that
$0 < p' \leq 1$,  for all $(s',p') \in \nxt(s)$ 
and
$ { \sum_{(s',p')\in \nxt(s)} p'= 1}$.
Function $\lab$ returns for each element of type \textsf{proc} a function associating a \textsf{bool} to each atomic proposition $a$ in $\textsf{lab}$.
Each instantiation of the algorithm consists in the appropriate definition
of $\nxt$ and $\lab$, depending on the language at hand and its semantics.

The local model-checking algorithm is defined as a function, \checkDtmc, 
shown in Table~\ref{alg:check_dtmc}.
On atomic state-formulae, the function returns the value
of $\lab$; when given a non-atomic state-formula, \checkDtmc{}  
calls itself recursively on sub-formulae, in case they are state-formulae, whereas it calls function 
\checkPath, in case the sub-formula is a path-formula.
In both cases the result is a Boolean value that indicates whether the state satisfies the formula.

\begin{table}[tbp]
\begin{lstlisting}
 $\checkDtmc$( $s: \mathsf{proc},$  $\Phi: \mathsf{formula}$)=
 match  $\Phi$
 with
 |$a\rightarrow (\lab~s~a)$
 |$\neg\Phi_1\rightarrow \neg\checkDtmc(s,\Phi_1)$
 |$\Phi_1\vee \Phi_2\rightarrow\checkDtmc(s,\phi_1)\vee \checkDtmc(s,\Phi_2)$ 
 |$\Prob{\relop}{p}{\varphi}\rightarrow \checkPath(s,\varphi)\relop p$
\end{lstlisting}
\caption{\label{alg:check_dtmc} Function $\checkDtmc$}
\end{table}

Function \checkPath, shown in Table~\ref{alg:checkpath_dtmc}, takes a state $s \in \textsf{proc}$ and a PCTL path-formula $\varphi \in \textsf{path\_formula}$ as input. As a result, it produces the probability measure of the set of paths, starting in state $s$, which satisfy path-formula $\varphi$. Following the definition of the formal semantics of PCTL, two 
different cases can be distinguished. If $\varphi$ has the form $\lnxt \Phi$ then the result is  the sum of the probabilities of the transitions from $s$ to those next states $s'$ that satisfy $\Phi$. To verify the latter, function \checkDtmc{} is recursively invoked on such states. If 
$\varphi$ has the form $\Phi_1 \rU{\le k}{} \Phi_2$ then we first check if $s$ satisfies $\Phi_2$, then $1$ is returned, since $\varphi$ is trivially satisfied. If $s$ does not satisfy $\Phi_1$ then $0$ is returned, since  $\varphi$ is trivially violated.  For the remaining case we need to recursively invoke \checkPath{} 
for the states reachable in one step from $s$, i.e. the states in the set 
$\{s'| \exists p': (s',p') \in \nxt(s)\}$. Note that these invocations of \checkPath{} are made on 
$\varphi' = \Phi_1 \rU{\le k-1}{} \Phi_2$ if $k>0$. If $k \leq 0$ 
then the formula is trivially not satisfied by $s$ and the value $0$ is returned.

\begin{table}[tbp]
\begin{lstlisting}
 $\checkPath$( $s: \mathsf{proc}$,  $\varphi: \mathsf{path\_formula}$ )=
 match $\varphi$ with
 |$\lnxt \Phi\rightarrow $ let $p$ = 0.0 and $\mathit{lst}$ = $\nxt(s)$ in
      for $(s',p')\in \mathit{lst}$ do if $\checkDtmc( s' , \Phi )$ then p $\leftarrow p+p'$
      done;
      $p$
 |$ \Phi_1 \rU{\le k}{} \Phi_2\rightarrow$ if $\checkDtmc( s , \Phi_2 )$ then 1.0
        else if $\checkDtmc( s , \neg \Phi_1 )$ then 0.0 
        else if $k>0$ then
               begin
                   let $p$ = 0.0 and $\mathit{lst}$ = $\nxt(s)$ in
                   for $(s',p')\in \mathit{lst}$ do
            	      $p\leftarrow p+p'*\checkPath(s', \Phi_1 \rU{\le k-1}{} \Phi_2)$
                   done;
                   $p$
               end
             else 0.0
\end{lstlisting}
\caption{\label{alg:checkpath_dtmc} Function $\checkPath$}
\end{table}

%
%

Let $s$ be a term of a probabilistic process language and $\calM$ 
the complete discrete time stochastic process associated with $s$ by the formal semantics of the language. 
The following theorem is easily proved by induction on $\Phi$~\cite{La+13b}.
\begin{theorem}
\label{LMC:CORR}
$s \models_{\calM}  \Phi$ if and only if $\checkDtmc(s,\Phi)=\mathsf{true}.\endstat$
\end{theorem}


\section{Modelling language}
\label{ModellingLanguage}

In this section we define a simple population description language. The language
is essentially a textual version of the graphical notation used in~\cite{BMM07}. A {\em system}
is defined as a population of $N$ identical interacting processes or objects\footnote{In~\cite{BMM07} {\em object}
is used instead of {\em process}. We consider the two terms synonyms here.}. At any
point in time, each object can be in any of its finitely many states and the evolution of the system
proceeds in a {\em clock-synchronous} fashion: at each clock tick each member 
of the population must either execute one of the transitions that are enabled in
its current state, or remain in such a state.\footnote{For the purpose of the present paper, language expressivity is not  a main concern.}

\paragraph{Syntax.}
Let $\act{}$ be a denumerable non-empty set of {\em actions}, ranged over by
$a, a', a_1, \ldots$ and $\sc{}$ be a denumerable non-empty set of {\em state constants}, ranged over by $C,C', C_1, \ldots$
An {\em object specification} $\os$ is a set $\SET{D_i}_{i\in I}$, for finite index set $I$, where each {\em state definition} $D_i$ has the form
$
C_i \dfas \sum_{j\in J_i} a_{ij}.C_{ij}
$,
with $J_i$ a finite index set, states $C_i,C_{ij} \in \sc{}$, and $a_{ij} \in \act{}$,
for $i\in I$ and $j\in J_i$. Intuitively, the notation $\sum_{j\in J_i} a_{ij}.C_{ij}$
is to be intended as the n-ary extension of the standard process algebraic
binary non-deterministic choice operator. We require that $a_{ij} \not= a_{ij'}$, for $j\not= j'$ and that for each 
state constant $C_{ij}$ occurring in the r.h.s. of a state definition $D_i$ of $\os$ 
there is a  {\em unique} $k \in I$ such that $C_{ij}$ is the l.h.s. of $D_k$. 

\begin{example}[An epidemic model~\cite{Bo+13}]
\label{ex:epidemicmodel}
We consider a network of computers that can be infected by a worm. Each node in the network can acquire infection 
from two sources, i.e. by the activity of a worm of an infected node 
({\tt inf\_sus})
or by an external source ({\tt inf\_ext}). Once a computer is infected, the worm 
remains latent for a while, and then activates ({\tt activate}). When the worm is active, it tries to propagate over the network by sending messages to other nodes. 
After some time, an infected computer can be patched ({\tt patch}), so that the infection is recovered. New versions of the worm can appear; for this reason, recovered 
computers can become susceptible to infection again, after a while ({\tt loss}). 
%
The object specification of the epidemic model is the following:\\
%
%

{\tt
S := inf\_ext.E + inf\_sus.E

E := activate.I

I := patch.R 

R := loss.S
}

\end{example}

The set of all actions occurring in object specification $\os$ is denoted by $\act{\os}$. Similarly, the set of states is denoted by $\sc{\os}$, ranged over by $c,c', c_1 \cdots$. 
In  Example~\ref{ex:epidemicmodel}, 
we have $\act{\mathtt{EM}} = \SET{\mathtt{inf\_ext}, \mathtt{inf\_sus}, \mathtt{activate}, \mathtt{patch}, \mathtt{loss} }$ and $\sc{\mathtt{EM}} = \SET{\mathtt{S},\mathtt{E}, \mathtt{I},\mathtt{R}}$.
A system is assumed composed of $N$ interacting instances of an object. 
Interaction among objects
is modelled probabilistically, as described below. Each action in $\act{\os}$
is assigned a probability value, that may depend on the global state of the system. This
is achieved by means of a {\em probability function definition}, that takes the
following form: $a :: E$, where $a \in \act{\os}$ and $E$ is an expression,
i.e. an element of $\mbox{Exp}$, defined according 
to the following grammar:
\[
E ::= v \, \sep \, \frc C \, \sep \, \uop \, E \, \sep \, E \, \bop \, E \, \sep \,(E)
\]
where  $v \in [0,1]$ and for each state $C$, $\frc C$ denotes the fraction of objects,
over the total number of objects $N$, in the system, that are currently in state $C$. Operators
$ \uop$ and $\bop$ are standard arithmetic unary and binary operators.

\begin{example}[Probability function definitions]
\label{ex:transition}
For the \emph{epidemic model} of Example~\ref{ex:epidemicmodel} we assign the following probability function definitions:\\[1em]
%
%

{\tt
inf\_ext :: $\alpha_{e}$;

inf\_sus :: $\alpha_{i} * (\frc I)$;

activate :: $\alpha_{a}$;

patch :: $\alpha_{r}$;

loss :: $\alpha_{s}$;
}\\

\noindent
where $\alpha_{e}$, $\alpha_{i}$, $\alpha_{a}$, $\alpha_{r}$ and $\alpha_{s}$ are model 
parameters in $[0,1]$, with $\alpha_{e} + \alpha_{i}  \leq 1$.
%
\end{example}

%
%
%
%

\noindent
A {\em system specification} is a triple 
$\langle \os, A, \vct{C_0} \rangle\aN$
where $\os$ is an object specification, $A$ is a set of probability function definitions
containing exactly one definition for each $a \in \act{\os}$, and 
$\vct{C_0}=\langle c_{0_1},\ldots,c_{0_N}  \rangle$ is the   {\em initial system state}, with 
$c_{0_n} \in \sc{\os}$, for $n=1\ldots N$; we say that $N$ is the  {\em population size}\footnote{Appropriate syntactical shorthands can be introduced for describing the initial state,
e.g. $\langle {\tt S[2000],E[100],I[200],R[0]} \rangle$ for 2000 objects initially in state 
{\tt S} etc.}; in the sequel, we will omit the explicit indication of the size $N$ in $\langle \os, A, \vct{C_0} \rangle\aN$,
and elements thereof or related functions,
writing simply $\langle \os, A, \vct{C_0} \rangle$, when this cannot cause confusion. 
\paragraph{Semantics.}
Let $\langle \os, A, \vct{C_0} \rangle$ be a system specification.
We associate with $\os$ the Labelled Transition System (LTS) 
$\langle \sc{\os}, \act{\os}, \trans{} \rangle$, where $\sc{\os}$ and $\act{\os}$
are the states and labels of the LTS, respectively, and the transition
relation $\trans{} \subseteq \sc{\os} \times  \act{\os} \times \sc{\os}$ 
is the smallest relation induced by rule 
(\ref{ObjLTS}).

\begin{equation}
\label{ObjLTS}
\sosrule
{C \dfas \sum_{j \in J} a_j.C_j \quad k \in J}
{C \trans{a_k} C_k}
\end{equation}

In the following we let 
$
\us{S} =\SET{\vct{m}\in [0,1]^S | \sum_{i=1}^{S} \vct{m}_{[i]} =1}
$
be the unit simplex of dimension $S$; furthermore, 
we let $c, c', C, C'  \ldots$ range over $\sc{\os}$ and for generic
vector $\vct{v} = \langle  v_1, \ldots, v_r \rangle$ we let
$\vct{v}_{[j]}$ denote the $j$-th component $v_j$ of $\vct{v}$, for $j=1,\ldots,r$.
A (system) {\em global state} is a tuple
$\vct{C}\aN \in  \sc{\os}^N$.  
W.l.g., we  assume that 
$\sc{\os} = \SET{C_1,\ldots,C_S}$
and that a total order is defined on state constants $C_1,\ldots,C_S$
so that we can unambiguously associate each component 
of a vector $\vct{m} = \langle  m_1, \ldots, m_S \rangle \in \us{S}$ with
a distinct element of $\SET{C_1,\ldots,C_S}$.
With each global state $\vct{C}\aN$ 
an {\em occupancy measure} vector
$\vct{M}\aN (\vct{C}\aN) \in \us{S}$ is associated where
$\vct{M}\aN (\vct{C}\aN) = \langle M\aN_1,\ldots,M\aN_S  \rangle$ 
with
$$
M\aN_i = \frac{1}{N}\sum_{n=1}^{N} \mathbf{1}_{\SET{\vct{C}\aN_{[n]} = C_i}}
$$
for $i=1,\ldots,S$, and the value of 
$\mathbf{1}_{\SET{\alpha = \beta}}$ is $1$, if 
$\alpha = \beta$, and $0$ otherwise.


A probability function definition $a:: E$ associates a real value
to action $a$ by evaluating $E$ in the current global state, via the interpretation function $\expsem$. In practice 
the occupancy measure representation of the state is used in $\expsem$.

 The expressions
interpretation function $\expsem : \mbox{Exp} \rightarrow \us{S} \rightarrow \reals$ is defined as usual: \\
%
%

$\semint{\expsem}{v}_{\vct{m}}  =  v$

$\semint{\expsem}{\frc C_i}_{\vct{m}} = \vct{m}_{[i]}$

$\semint{\expsem}{\uop \, E}_{\vct{m}}  =  \uop \, (\semint{\expsem}{E}_{\vct{m}})$

$\semint{\expsem}{E_1 \, \bop \, E_2}_{\vct{m}}  = (\semint{\expsem}{E_1}_{\vct{m}}) \, \bop \, (\semint{\expsem}{E_2}_{\vct{m}})$

$\semint{\expsem}{(E)}_{\vct{m}}  = (\semint{\expsem}{(E)}_{\vct{m}})$\\

\noindent
The set $A$ of probability function definitions characterises a function 
$\pi$ with type $\us{S}  \rightarrow  \act{\os} \rightarrow \reals$ as follows:
for each $a::E$ in $A$, we have
$\pi(\vct{m}, a) = \semint{\expsem}{E}_{\vct{m}}$.

For a system specification of size $N$, we define the {\em object transition matrix} as follows: 
$\otm\aN: \us{S} \times \sc{\os} \times \sc{\os}  \rightarrow \reals$, with
\begin{equation}
\label{OTM}
\otm\aN(\vct{m})_{c,c'} =
\left\{
\begin{array}{l}
\sum_{a: c \trans{a} c' } \pi(\vct{m}, a), \mbox{ if } c\not = c',\\[.20cm]
1 - \sum_{a \in I(c)} \pi(\vct{m}, a), \mbox{ if } c = c'.
\end{array}
\right.
\end{equation}
where $I(c) = \SET{a \in \act{\os}| \exists c' \in  \sc{\os}: c\trans{a}c' \not=c}$. 
We say that a  state $c \in \sc{\os}$ is {\em probabilistic} in $\vct{m}$ 
if 
$
0 \leq \sum_{a \in {I^*(c)}} \pi(\vct{m}, a) \leq 1
$
where set $I^*(c)$ is defined as follows: $I^*(c) = I(c) \cup \SET{a \in \act{\os} | c\trans{a}c}$.
Note that whenever all states in $\sc{\os}$ are probabilistic in $\vct{m}$, matrix
$\otm\aN(\vct{m})$ is a one step transition probability matrix.
We define the (system) {\em  global state  transition matrix}  
$\stm\aN: \us{S} \times \sc{\os}^N \times \sc{\os}^N  \rightarrow \reals$, as 
$$
\stm\aN(\vct{m})_{\vct{C},\vct{C}'} = \Pi_{n=1}^{N} \otm\aN(\vct{m})_{\vct{C}_{[n]},\vct{C}_{[n]}'}.
$$
Note that whenever all states in $\sc{\os}$ are probabilistic in $\vct{m}$, matrix
$\stm^N(\vct{m})$ is a one step transition probability matrix modelling
a possible single step of the system as result of the parallel
execution of a single step of each of the $N$ instances of the object.
In this case, the $S^N \times S^N$ matrix $\sdtm\aN$
with 
\begin{equation}
\label{SyDTMC}
\sdtm\aN_{\vct{C},\vct{C}'} =
\stm\aN(\vct{M}\aN (\vct{C}))_{\vct{C},\vct{C}'}
\end{equation}
is  the one-step transition matrix of a (finite state) DTMC, namely the DTMC of the
system composed on $N$ objects specified by $\os$.  In this case, 
we let $\vct{X}\aN(t)$ denote
the Markov process with transition probability matrix $\sdtm\aN$ as above
and $\vct{X}\aN(0)= \vct{C_0\aN}$, i.e. with initial probability distribution  
$\dirac_{\vct{C_0\aN}}$, where $\vct{C_0\aN}$
is the initial system state and $\dirac_{\vct{C_0\aN}}$ is the Dirac distribution
with the total mass on $\vct{C_0\aN}$. With a little bit of notational overloading, 
we define the `occupancy measure DTMC' as 
 $\vct{M}\aN (t) = \vct{M}\aN (\vct{X}\aN(t))$; for $\vct{m} = \vct{M}\aN(\vct{C})$, for some state $\vct{C}$ of DTMC $\vct{X}(t)$, we have: 
\begin{equation}
\label{PRMX:OCC:MEA}
\Pr\SET{\vct{M}\aN (t+1) = \vct{m}' \mid \vct{M}\aN (t) = \vct{m}}
= \sum_{\vct{C}': \vct{M}\aN(\vct{C}')=\vct{m}'} \sdtm\aN_{\vct{C},\vct{C}'}
\end{equation}
Note that the above definition is a good definition; in fact, 
if $\vct{M}\aN(\vct{C}) = \vct{M}\aN(\vct{C}'')$, then 
$\vct{C}$ and $\vct{C}''$ are just two permutations of the same
local states. This implies that for all $\vct{C}'$  we have
$\sdtm\aN_{\vct{C},\vct{C}'} = \sdtm\aN_{\vct{C}'',\vct{C}'}$.

\paragraph{PCTL local Model-checking.}
For the purpose of expressing system properties in PCTL, 
we partition the set of atomic propositions $\sls$ into sets $\sls_1$
and $\sls_g$. 
Given system specification $\langle \os, A, \vct{C_0\aN} \rangle\aN$, we extend it 
with a {\em state labelling function definition} that associates each state 
$c\in \sc{\os}$ with a (possibly empty) finite set $\slf_1(c)$ of propositions from  $\sls_1$.
We extend $\slf_1$ to global states with $\slf_1(\langle c_1, \ldots, c_N \rangle) = \slf_1(c_1)$; this way, we can express {\em local} properties of the first object in the system, in the {\em context} of the complete 
population\footnote{Of course, the choice of the {\em first} object is purely conventional. Furthermore,
all the results which in the present paper are stated w.r.t. the {\em first} object of a system, are easily extened to
finite subsets of objects in the system. For the sake of notation, in the rest of the paper, we stick 
to the {\em first object} convention.}. 
In order to express also (a limited class of) {\em global} properties of the population, we use set $\sls_g$.
The  system specification is further enriched by associating labels $a \in \sls_g$  with  expressions $\mbox{bexp}$
in the class $\mbox{BExp}$ of restricted boolean expressions. We assume a sublanguage of function specifications
be given\footnote{The specific features of the sublanguage are not relevant for the purposes of the present paper and we leave their treatment out for the sake of simplicity.} and for function symbol $F$, 
$\semint{\expsem}{F}_{\vct{m}}: [0,1]^q \mapsto \reals$  continuous in $[0,1]^q$,
with $\semint{\expsem}{F}_{\vct{m}} = \semint{\expsem}{F}_{\vct{m}'}$ 
for all $\vct{m},\vct{m}' \in \us{S}$; then $\mbox{BExp}$ is the set of expressions
of the form
$F(E_1,\ldots,E_q) \, \relop \, r$, where  
each $E_j$ is of the form $\frc C$, $\relop \in \SET{>,<}$, $r \in \reals$ and
$
\semint{\expsem}{F(E_1,\ldots,E_q)}_{\vct{m}} =
\semint{\expsem}{F}_{\vct{m}}(\semint{\expsem}{E_1}_{\vct{m}},\ldots,\semint{\expsem}{E_q}_{\vct{m}})
$.

We define the
state global labelling function $\slf_g$ as 
$$
\slf_g(\langle c_1, \ldots, c_N \rangle)=
\SET{a  \in \sls_g \mid 
\semint{\expsem}{\mbox{bexp}_a}_{\vct{M}\aN(\vct{\langle c_1, \ldots, c_N \rangle})}  =  \true}.
$$
We obtain the state labelled DTMC $\calD\aN(t)$ from
 $\vct{X}\aN(t)$, with
transition matrix $\sdtm\aN$ above, by enriching it
with labelling function $\slf_{\calD\aN}$ such that 
$\slf_{\calD\aN}(\vct{C}) = \slf_1(\vct{C}) \cup \slf_g(\vct{C})$.

The definition of $\Paths_{\calD\aN}(\vct{C}\aN)$ as well as that of the satisfaction
relation $\models_{\calD\aN}$ are obtained by instantiating those given in Sect.~\ref{Logic} to  $\calD\aN$.
For $\sigma \in  \Paths_{\calD\aN}(\vct{C}\aN)$, $\sigma[j]_{[n]}$ denotes
the $n$-th local state of global state $\sigma[j]$. \\
For model-checking a system specification  $\langle \os, A, \vct{C_0\aN} \rangle\aN$ 
we instantiate \textsf{proc} with\footnote{Strictly speaking, the relevant
components of the algorithm are instantiated to {\em representations} of the terms, sets
and functions mentioned in this section. For the sake of notational simplicity, we often use the same notation both for mathematical objects and for their representations.} 
$\sc{\os}^N$
and \textsf{lab} with  $\sls_1 \cup \sls_g$.
Function $\nxt$ is instantiated to the function $\nxt_{\calD\aN}$, where
$$
\nxt_{\calD\aN}(\vct{C}) = 
[(\vct{C}',p') \mid \sdtm\aN_{\vct{C},\vct{C}'} = p' > 0 ].
$$
Given a vector $\vct{C}$, $\nxt_{\calD\aN}(\vct{C})$
computes a list corresponding to the positive elements of the row of matrix $\sdtm\aN$ associated with $\vct{C}$. Of course, only those elements of $\sdtm\aN$ 
that are necessary for $\nxt_{\calD\aN}$ are actually computed. 
Function 
$\lab$ is instantiated with the function $\lab_{\calD\aN}: 
\sc{\os}^N \times \act{\os} \rightarrow \mathbb{B}$ with
$\lab_{\calD\aN}(\vct{C},a)= a \in \slf_{\calD\aN}(\vct{C})$.

\begin{example}[Properties]
\label{ex:properties} 
For the epidemic model of Example~\ref{ex:epidemicmodel} we can consider the following properties, where 
$i, e, r \in \sls_1$ are labelling states $I$, $E$ and $R$, respectively, and 
$\mathit{LowInf} \in \sls_g$ is defined as $(\frc I )< 0.25$:
\begin{enumerate} 
\item[P1] 
the worm will be active in the first component within  $k$ steps with a probability that is at most $p$:
$
\Prob{\leq}{p}{~true~\rU{\le k}~i~}
$;
\item[P2] the probability that the first component is infected, but latent, in the next $k$ steps while the worm is active on less then $25\%$ of the components is
at most $p$:
$
\Prob{\leq}{p}{\mathit{LowInf} \rU{\le k}~e~}
$;
\item[P3] \label{P3} the probability to reach, within $k$ steps, a configuration where the first component is not infected but the worm will be
activated with probability greater than $0.3$ within $5$ steps
is at most $p$:\\[0.5em]
$\mbox{ }$\hspace{0.7in}
$
\Prob{\leq}{p}{~true~\rU{\le k} ( !e \wedge !i\wedge \Prob{>}{0.3}{~true~\rU{\le 5}~i~} )}.
$
\end{enumerate}
\end{example}
In Fig.~\ref{tab:completemc} the result of exact PCTL model-checking of Ex.~\ref{ex:epidemicmodel} is reported.
On the left the probability of the set of paths that satisfy the path-formulae used in the three formulae above is shown for a system composed of eight objects each in initial state $S$, for $k$ from 0 to 70.
%
On the right the time needed to perform the analysis using PRISM~\cite{KNP04} and using exact on-the-fly PCTL model checking are presented\footnote{We use a $1.86 GHz$ Intel Core 2 Duo with 4 GB. State space generation time of PRISM is not counted. The experiments are available at \url{http://rap.dsi.unifi.it/$\sim$loreti/OFPMC/}).}, showing that the latter has comparable performance. Worst-case complexity of both algorithms are also comparable.
%
\begin{figure}[tbp]
\begin{tabular}{cc}
\begin{minipage}{7cm}
\includegraphics[scale=0.25]{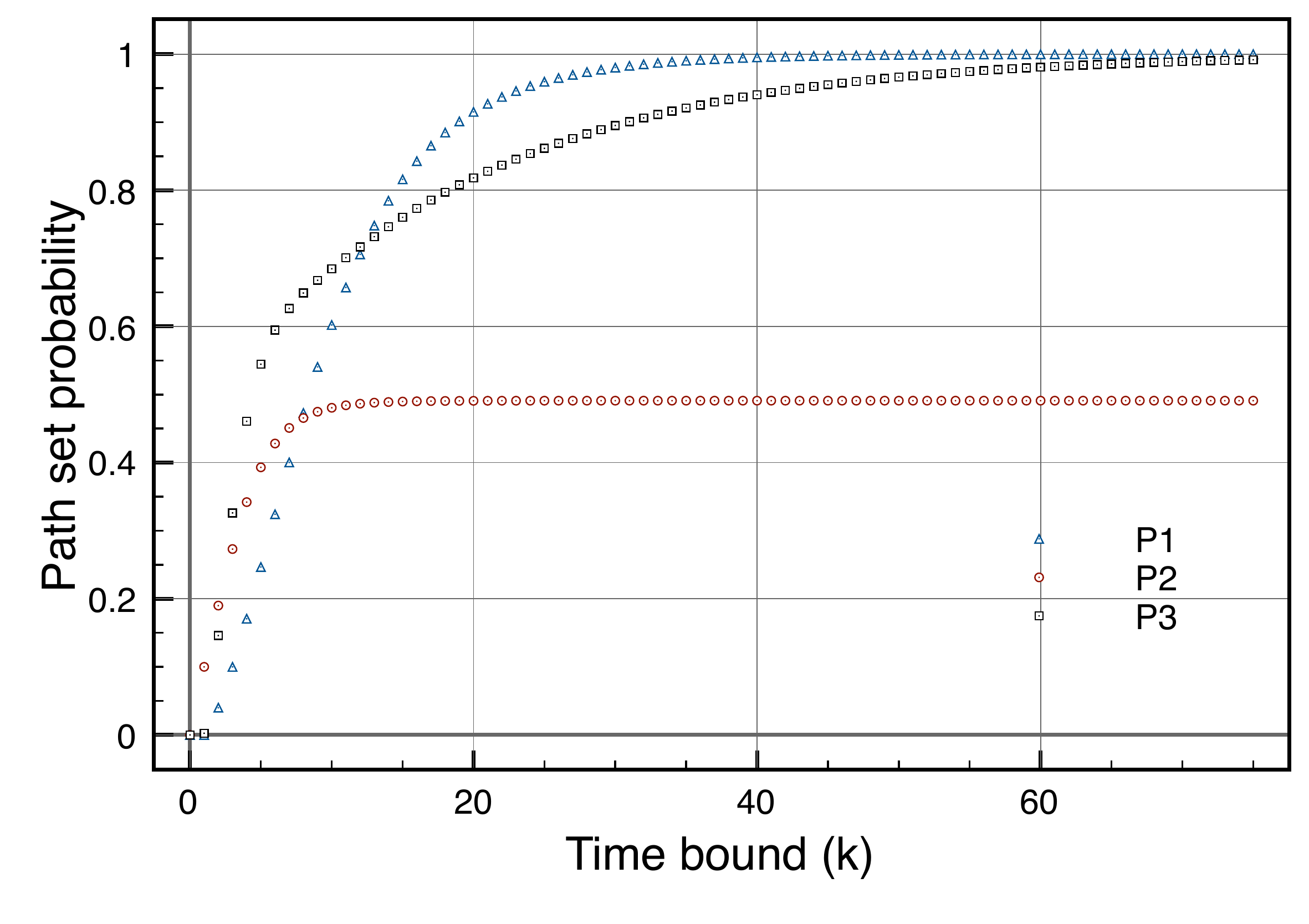} 
\end{minipage} &

\begin{tabular}{c|c|c}
         & \quad PRISM \quad & \quad Exact on-the-fly \\
\hline     
$P1$         &  $108.479s$ & $29.587s$ \\
$P2$         &  $51.816s$ & $3.409s$ \\
$P3$         &  $216.952s$ & $85.579s$\\
\multicolumn{3}{l}{}\\
\multicolumn{3}{l}{Model parameter values:}\\
\multicolumn{3}{l}{$\alpha_e=0.1$, $\alpha_i=0.2$, $\alpha_r=0.2$}\\
\multicolumn{3}{l}{$\alpha_a=0.4$, $\alpha_s=0.1$}
\end{tabular} 
\end{tabular}

\caption{Exact model-checking results (left) and verification time (right).}
\label{tab:completemc}
\end{figure}
The local model-checker has been instantiated with the model 
defined by the (exact) operational semantics of the language, where each state 
$\vct{C} \in \sc{\os}^N$ is a {\em global} system state. In Sect.~\ref{MeanFieldAprx}
we instantiate the procedure with the mean-field, approximated, semantics of the language,
leading to a scalable, `fast', model-checker, insensitive to the population size. 


\section{Fast Mean-field Model-checking}
\label{MeanFieldAprx}

Given a  system specification $\langle \os, A, \vct{C_0\aN} \rangle\aN$ 
with 
initial state $\vct{C_0}$,
we want to focus on the behaviour of the first object, starting in the initial 
state $\vct{C_0}_{[1]}$, when in execution with all the other
objects for (very) large population size $N$. We define a mapping
$\map\aN: \sc{\os}^N \rightarrow (\sc{\os} \times \us{S})$ such that
$\map\aN(\vct{C}\aN) = \langle \vct{C}\aN_{[1]},\vct{M}\aN (\vct{C}\aN)\rangle$.
Note that $\map\aN$ and $\calD\aN(t)$ together define a state labelled DTMC, denoted $\map\calD\aN(t)$, and defined as $\map\aN(\vct{X}\aN(t))$,
with 
$
\slf_1(\langle c ,\vct{m} \rangle) = \slf_1(c), \, 
\slf_g(\langle c ,\vct{m} \rangle)=
\SET{a  \in \sls_g \mid 
\semint{\expsem}{\mbox{bexp}_a}_{\vct{m}}  =  \true}
$, and
$\slf_{\map\calD\aN}(\langle c ,\vct{m} \rangle)$ defined as 
$\slf_1(\langle c ,\vct{m} \rangle) \,\cup \, \slf_g(\langle c ,\vct{m} \rangle)
$, where $\sls_1, \sls_g$ and $\mbox{bexp}_a$ are defined
in a similar way as in Sect.~\ref{ModellingLanguage}. 
The one-step matrix 
of $\map\calD\aN(t)$ is:
\begin{equation}
\label{PRMX:H}
\hdtm\aN_{\langle c, \vct{m}\rangle, \langle c', \vct{m}' \rangle} =
\sum_{\vct{C}': \map\aN(\vct{C}') =  \langle c', \vct{m}' \rangle} 
\sdtm\aN_{\vct{C},\vct{C}'}
\end{equation}
where $\vct{C}$ is such that 
$\map\aN(\vct{C}) = \langle c, \vct{m}\rangle$.\footnote{With a similar 
argument as for definition (\ref{PRMX:OCC:MEA}), noting  that 
$\vct{M}\aN(\vct{C}) = \vct{M}\aN(\vct{C}'')$ and $\vct{C}_{[1]} = \vct{C}''_{[1]}$,
it can be easily seen that also definition (\ref{PRMX:H}) is a good definition.}
The definitions of paths  for state $\langle c,\vct{m}\rangle$ of $\map\calD\aN$,  
$\Paths_{\map\calD\aN}(\langle c,\vct{m}\rangle)$, 
of $\at_{\map\calD\aN}(t)$ and of the
satisfaction relation $\models_{\map\calD\aN}$
of PCTL formulas against  $\map\calD\aN(t)$, are obtained by instantiating the relevant definitions of Sect.~\ref{Logic}
to the model $\map\calD\aN(t)$.
Furthermore, we let 
$
\at_{\map\calD\aN}(t,c) = 
\SET{\langle c',\vct{m}'\rangle \in \at_{\map\calD\aN}(t) \mid c' = c}
$.

We extend mapping $\map\aN$ to 
sets and
paths in the obvious way:  for set $X$ of states, let
$\map\aN(X) = \{\map\aN(x) \mid x \in X\}$, and 
for 
$\sigma \in \Paths_{\calD\aN}(\vct{C\aN})$, let
$\map\aN(\sigma) = \map\aN(\sigma[0])\map\aN(\sigma[1])\map\aN(\sigma[2])\cdots$

\ifTR%
Note that, by definition of $\map\aN$ and $\map\calD\aN(t)$, if 
$\sigma \in \Paths_{\calD\aN}(\vct{C\aN})$, then  
$\map\aN(\sigma) \in \Paths_{\map\calD\aN}(\map\aN(\vct{C\aN}))$ and 
if $\rho \in \Paths_{\map\calD\aN}(\langle c, \vct{m} \rangle)$, then there exist
$\vct{C\aN}$, with  $\map\aN(\vct{C\aN}) = \langle c, \vct{m} \rangle$, and 
$\sigma \in \Paths_{\calD\aN}(\vct{C\aN})$ s.t.
$\map\aN(\sigma) = \rho$.
Furthermore, it is easy to see that
$
\map\aN(\at_{\calD\aN}(t)) = \at_{\map\calD\aN}(t)
$
and $\vct{C} \in \at_{\calD\aN}(t)$ iff $\map\aN(\vct{C}) \in \at_{\map\calD\aN}(t)$. 
\fi%
\ifTR
The following lemma relates the two interpretations of the logic.
\else
The following lemma relates the two interpretations of the logic, and can be easily 
proved by induction on formulae $\Phi$~\cite{La+13a}.
\fi

\begin{lemma}
\label{CorrAbs}
For all $N > 0$, states $\vct{C}\aN$ and formulas $\Phi$ the following holds:
$
\vct{C}\aN \models_{\calD\aN} \Phi \mbox{ iff }
\map\aN(\vct{C}\aN) \models_{\map\calD\aN} \Phi.
\endstat
$
\end{lemma}

\ifTR
\begin{proof}
By induction on $\Phi$; in the proof we write $\vct{C}$ instead of $\vct{C}\aN$
for the sake of readability. \\[1em]

\noindent
{\bf Case $a$}:\\
$\vct{C} \models_{\calD\aN} a$ if and only if 
$
a \in \slf_1(C_{[1]})\,  \RED{\cup\, \slf_g(\vct{C})}
$, 
by definition of $\models_{\calD\aN}$.\\
$\map\aN(\vct{C}) = \langle \vct{C}_{[1]},\vct{M}\aN (\vct{C})\rangle$, by definition of $\map\aN$.\\
Clearly, \RED{if $a\in \slf_1(C_{[1]})$, then} $\vct{C} \models_{\calD\aN} a$ if and only if 
$\map\aN(\vct{C}) \models_{\map\calD\aN} a$, by definition of $\models_{\map\calD\aN}$.
\RED{If $a\in \slf_g(\vct{C})$, then $\vct{C} \models_{\calD\aN} a$ if and only if 
$\map\aN(\vct{C}) \models_{\map\calD\aN} a$, by definition of $\models_{\map\calD\aN}$
and of $\semint{\expsem}{\mbox{bexp}_a}$.}\\[1em]

\noindent
{\bf Case $\neg \Phi$}:\\
$\vct{C} \models_{\calD\aN} \neg \Phi$ if and only if  
$\mbox{not } \vct{C} \models_{\calD\aN} \Phi$, by definition of $\models_{\calD\aN}$.
Thus, 
$\mbox{not } \map\aN(\vct{C}) \models_{\map\calD\aN} \Phi$, by the induction hypothesis, and
$\map\aN(\vct{C}) \models_{\map\calD\aN} \neg \Phi$, by definition of $\models_{\map\aN}$.\\[1em]

\noindent
{\bf Case $\Phi_1 \vee \, \Phi_2$}:\\
$\vct{C} \models_{\calD\aN} \Phi_1 \vee \, \Phi_2$ if and only if  
$\vct{C} \models_{\calD\aN} \Phi_1$ or 
$\vct{C} \models_{\calD\aN} \Phi_2$, by definition of $\models_{\calD\aN}$.
Thus, 
$\map\aN(\vct{C}) \models_{\map\calD\aN} \Phi_1$ or
$\map\aN(\vct{C}) \models_{\map\calD\aN} \Phi_2$, by the induction hypothesis, and
$\map\aN(\vct{C}) \models_{\map\calD\aN} \Phi_1 \vee \, \Phi_2$, by definition of $\models_{\map\calD\aN}$.\\[1em]

\noindent
{\bf Case $\Prob{\bowtie}{p}{\lnxt \, \Phi}$}:\\

\noindent
$
\deriv
\Pr \{ \rho \in \Paths_{\map\calD\aN}(\map\aN(\vct{C})) \mid \rho \models_{\map\calD\aN} \nxt \, \Phi \}
\hint{=}{Def. $\rho \models_{\map\calD\aN} \lnxt \, \Phi$}
\sum_{\langle c', \vct{m}' \rangle: \langle c', \vct{m}' \rangle \models_{\map\calD\aN} \Phi}\hdtm\aN_{\map\aN(\vct{C}),\langle c', \vct{m}' \rangle}
\hint{=}{Def. $\hdtm\aN$}
\sum_{\langle c', \vct{m}' \rangle: \langle c', \vct{m}' \rangle \models_{\map\calD\aN} \Phi}
\sum_{\vct{C}': \map\aN(\vct{C}') =  \langle c', \vct{m}' \rangle}
\sdtm\aN_{\vct{C},\vct{C}'}
\hint{=}{I.H., i.e. $\map\aN(\vct{C}') \models_{\map\calD\aN} \Phi$ iff 
$\vct{C}' \models_{\calD\aN} \Phi$}
\sum_{\vct{C}' : \vct{C}' \models_{\calD\aN} \Phi}\sdtm\aN_{\vct{C},\vct{C}'}
\hint{=}{Def. $\sigma \models_{\calD\aN} \lnxt \, \Phi$}
\Pr \{ \sigma \in \Paths_{\calD\aN}(\vct{C}) \mid \sigma \models_{\calD\aN} \nxt \, \Phi \}
$

\noindent
{\bf Case $\Prob{\bowtie}{p}{\Phi_1 \rU{\le k}{} \Phi_2}$}:\\
\noindent
We prove the assert by (nested) induction on $k$.\\[1em]
{\bf Base case ($k=0$)}:\\
$
\deriv
\Pr \{ \rho \in \Paths_{\map\calD\aN}(\map\aN(\vct{C})) \mid 
\rho \models_{\map\calD\aN} \, \Phi_1 \rU{\le 0}{} \Phi_2\}
\hint{=}{Def. of $\rho \models_{\map\calD\aN} \, \Phi_1 \rU{\le k}{} \Phi_2$}
\Pr \{ \rho \in \Paths_{\map\calD\aN}(\map\aN(\vct{C})) \mid 
\rho[0] \models_{\map\calD\aN}  \Phi_2\}
\hint{=}{Def. of $\Paths_{\map\calD\aN}(\map\aN(\vct{C}))$ and $\rho[0]$}
\left\{
\begin{array}{l}
1, \mbox{ if } \map\aN(\vct{C})  \models_{\map\calD\aN}  \Phi_2,\\\\
0, \mbox{ if } \mbox{ not }\map\aN(\vct{C})  \models_{\map\calD\aN}  \Phi_2
\end{array}
\right.
\hint{=}{I. H. on logic formulas}
\left\{
\begin{array}{l}
1, \mbox{ if } \vct{C}  \models_{\calD\aN}  \Phi_2,\\\\
0, \mbox{ if } \mbox{ not }\vct{C}  \models_{\calD\aN}  \Phi_2
\end{array}
\right.
\hint{=}{Def. of $\Paths_{\calD\aN}(\vct{C})$ and $\sigma[0]$}
\Pr \{ \sigma \in \Paths_{\calD\aN}(\vct{C}) \mid 
\sigma[0] \models_{\calD\aN}  \Phi_2\}
\hint{=}{Def. of $\sigma \models_{\calD\aN} \, \Phi_1 \rU{\le k}{} \Phi_2$}
\Pr \{ \sigma \in \Paths_{\calD\aN}(\vct{C}) \mid 
\sigma \models_{\calD\aN} \, \Phi_1 \rU{\le 0}{} \Phi_2\}
$\\[1em]
\noindent
{\bf Induction step}:\\

\noindent
$
\deriv
\Pr \{ \rho \in \Paths_{\map\calD\aN}(\map\aN(\vct{C})) \mid 
\rho \models_{\map\calD\aN} \, \Phi_1 \rU{\le k+1}{} \Phi_2\}
\hint{=}{Def. $\rho \models_{\map\calD\aN} \, \Phi_1 \rU{\le k+1}{} \Phi_2$}
\left\{
\begin{array}{l}
0, \mbox{ if not } \map\aN(\vct{C})  \models_{\map\calD\aN}  \Phi_1 \mbox{ and not }
\map\aN(\vct{C})  \models_{\map\calD\aN}  \Phi_2\\\\
1, \mbox{ if } \map\aN(\vct{C})  \models_{\map\calD\aN}  \Phi_2,\\\\
\sum_{\rho[1]: \rho[1] \models_{\map\calD\aN} \Phi_1} 
\hdtm\aN_{\map\aN(\vct{C}),\rho[1]} \cdot
\Pr \{ \rho' \in \Paths_{\map\calD\aN}(\rho[1]) \mid 
\rho' \models_{\map\calD\aN} \, \Phi_1 \rU{\le k}{} \Phi_2\}
\end{array}
\right.
\hint{=}{Def. $\hdtm\aN$}
\left\{
\begin{array}{l}
0, \mbox{ if not } \map\aN(\vct{C})  \models_{\map\calD\aN}  \Phi_1 \mbox{ and not }
\map\aN(\vct{C})  \models_{\map\calD\aN}  \Phi_2,\\\\
1, \mbox{ if } \map\aN(\vct{C})  \models_{\map\calD\aN}  \Phi_2\\\\
\sum_{\rho[1]: \rho[1] \models_{\map\calD\aN} \Phi_1}\\\hspace{0.2in}
\sum_{\vct{C}': \map\aN(\vct{C}') = \rho[1]}
\sdtm\aN_{\vct{C},\vct{C}'} \cdot
\Pr \{ \rho' \in \Paths_{\map\calD\aN}(\map\aN(\vct{C}')) \mid 
\rho' \models_{\map\calD\aN} \, \Phi_1 \rU{\le k}{} \Phi_2\}
\end{array}
\right.
\hint{=}{I.H. on $k$}
\left\{
\begin{array}{l}
0, \mbox{ if not } \map\aN(\vct{C})  \models_{\map\calD\aN}  \Phi_1 \mbox{ and not }
\map\aN(\vct{C})  \models_{\map\calD\aN}  \Phi_2,\\\\
1, \mbox{ if } \map\aN(\vct{C})  \models_{\map\calD\aN}  \Phi_2\\\\
\sum_{\rho[1]: \rho[1] \models_{\map\calD\aN} \Phi_1}\\\hspace{0.2in}
\sum_{\vct{C}': \map\aN(\vct{C}') = \rho[1]}
\sdtm\aN_{\vct{C},\vct{C}'} \cdot
\Pr \{ \sigma' \in \Paths_{\calD\aN}(\vct{C}') \mid 
\sigma' \models_{\calD\aN} \, \Phi_1 \rU{\le k}{} \Phi_2\}
\end{array}
\right.
\hint{=}{I.H. on logic formulas}
\left\{
\begin{array}{l}
0,  \mbox{ if not } \vct{C}  \models_{\calD\aN}  \Phi_1 \mbox{ and not }
 \vct{C}  \models_{\calD\aN}  \Phi_2\\\\
1, \mbox{ if } \vct{C}  \models_{\calD\aN}  \Phi_2\\\\
\sum_{\vct{C}': \vct{C}' \models_{\calD\aN} \Phi_1}
\sdtm\aN_{\vct{C},\vct{C}'} \cdot
\Pr \{ \sigma' \in \Paths_{\calD\aN}(\vct{C}') \mid 
\sigma' \models_{\calD\aN} \, \Phi_1 \rU{\le k}{} \Phi_2\}
\end{array}
\right.
\hint{=}{Def. $\sigma \models_{\calD\aN} \, \Phi_1 \rU{\le k+1}{} \Phi_2$}
\Pr \{ \sigma \in \Paths_{\calD\aN}(\vct{C}) \mid 
\sigma \models_{\calD\aN} \, \Phi_1 \rU{\le k+1}{} \Phi_2\}
$\hfill$\Box$

\end{proof}
\fi

%

We now consider the stochastic process $\map\calD(t)$ defined
below, for $c_0, c,c' \in \sc{\os}$, $\boldsymbol{\mu}_0, \vct{m}, \vct{m}' \in \us{S}$  
and function  $\otm(\vct{m})_{c,c'}$, continuous in $\vct{m}$:


\begin{equation}
\label{Def:HD}
\begin{array}{l}
\Pr\SET{\map\calD (0) = \langle c, \vct{m} \rangle} =
\dirac_{\langle c_0, \boldsymbol{\mu}_0 \rangle}(\langle c, \vct{m} \rangle),\\
\Pr\SET{\map\calD (t+1) \!= \!\langle c', \vct{m}' \rangle \!\mid \!\map\calD (t) =
\langle c, \vct{m} \rangle} \!= \hspace{-0.04in} 
\left\{
\begin{array}{ll}
\otm(\vct{m})_{c, c'}, &\mbox{if } \vct{m}' = \vct{m} \cdot \otm(\vct{m})\\
0, &\mbox{otherwise}.
\end{array}
\right.
\end{array}
\end{equation}

\noindent
The definition of the labeling function $\slf_{\map\calD}$ is the same as that of
$\slf_{\map\calD\aN}$. 
Note that $\map\calD$ is a DTMC with initial state $\langle c_0, \boldsymbol{\mu}_0 \rangle$; 
memoryless-ness as well as time homogeneity directly
follow from the definition of the process (\ref{Def:HD}).
The definitions of paths  for state $\langle c,\vct{m}\rangle$ of $\map\calD$,  
$\Paths_{\map\calD}(\langle c,\vct{m}\rangle)$, 
of $\at_{\map\calD}(t)$ and of the
satisfaction relation $\models_{\map\calD}$
of PCTL formulas against  $\map\calD(t)$ are obtained by instantiating the relevant definitions of Sect.~\ref{Logic}
to the model $\map\calD(t)$.
Furthermore, define function $\boldsymbol{\mu}(t)$ as follows:
$\boldsymbol{\mu}(0) = \boldsymbol{\mu}_0$ and 
$\boldsymbol{\mu}(t+1) = \boldsymbol{\mu}(t) \cdot \otm(\boldsymbol{\mu}(t))$;
then, for $t \geq 0$ and for $\langle c,\vct{m}\rangle \in \at_{\map\calD}(t)$
we have $\vct{m} = \boldsymbol{\mu}(t)$.

In the following we use the fundamental result stated below, due to Le Boudec et al.~\cite{BMM07}.
\ifTR%
We recall that, for each $N$, the occupancy measure process 
$\vct{M}\aN(t)$ is the stochastic process $\vct{M}\aN(t)= \vct{M}\aN(\vct{X}\aN(t))$,
with initial distribution $\dirac_{\vct{M}\aN(\vct{C_0\aN})}$ where $\vct{C_0\aN}$
is the initial state vector of the given system specification.
\fi%

\begin{quote}
\noindent\\
{\bf Theorem 4.1 of~\cite{BMM07}}
{\em
Assume that for all $c,c' \in \sc{\os}$, there 
exists  function  $\otm(\vct{m})_{c,c'}$, continuous in $\vct{m}$, such that,
for $N \rightarrow \infty$, $\otm\aN(\vct{m})_{c,c'}$ converges uniformly
in $\vct{m}$ to $\otm(\vct{m})_{c,c'}$.
Assume, furthermore, that there exists $\boldsymbol{\mu}_0 \in  \us{S}$ such that 
$\vct{M}\aN (\vct{C_0\aN})$ converges almost surely to $\boldsymbol{\mu}_0$.
Define function $\boldsymbol{\mu}(t)$ of $t$ as follows:
$\boldsymbol{\mu}(0) = \boldsymbol{\mu}_0$ and 
$\boldsymbol{\mu}(t+1) = \boldsymbol{\mu}(t) \cdot \otm(\boldsymbol{\mu}(t))$.
Then, for any fixed $t$, almost surely
$
\lim_{N\rightarrow\infty} \vct{M}\aN(t) = \boldsymbol{\mu}(t).
\endstat
$
}\\
\end{quote}

\begin{remark}
\label{GConvergence}
We observe that, as direct consequence of Theorem 4.1 of~\cite{BMM07}
and of the restrictions on the definition of $\mbox{BExp}$, 
for any fixed $t$ and for all $\epsilon >0$, there exists $\bar{N}$ such that,
for all $N\geq\bar{N}$, almost surely 
$$
\mid \semint{\expsem}{\mbox{bexp}}_{\vct{m}} -
\semint{\expsem}{\mbox{bexp}}_{\boldsymbol{\mu}(t)} \mid\, < \epsilon
$$
for all $\langle c,\vct{m}\rangle \in \at_{\map\calD\aN}(t)$ and $\mbox{bexp} \in \mbox{BExp}$.
In other words, for $N$ large enough and  $\langle c,\vct{m}\rangle \in \at_{\map\calD\aN}(t)$,
$\slf_g(\langle c,\vct{m}\rangle)= \slf_g(\langle c,\boldsymbol{\mu}(t)\rangle)$, and, consequently,
$\slf(\langle c,\vct{m}\rangle)= \slf(\langle c,\boldsymbol{\mu}(t)\rangle).\endstat$
\end{remark}
In the rest of the paper we will  focus on sequences
$\left(\langle \os, A, \vct{C_0} \rangle\aN\right)_{N\geq N_0}$ of system specifications, for some $N_0 > 0$. 
In particular, we will consider only sequences  $\left(\map\calD\aN(t)\right)_{N\geq N_0}$
such that for all $N\geq N_0$, $\vct{C_0}^{(N)}_{[1]}=\vct{C_0}^{(N_0)}_{[1]}$;
in other words we want the population size increase with $N$, while the (initial state of the) first 
object of the system is left unchanged.

Let us now go back to process $\map\calD(t)$, where, in equation (\ref{Def:HD}) we use function 
$\otm(\vct{m})_{c,c'}$ of the hypothesis of the theorem recalled above; 
similarly, 
for the initial distribution we use $\dirac_{\langle \vct{C_{0[1]}\aN}, \boldsymbol{\mu}(0) \rangle}$.

The following is a corollary of Theorem 4.1 and Theorem 5.1 (Fast simulation) presented in~\cite{BMM07},
when considering sequences $\left(\map\calD\aN(t)\right)_{N\geq N_0}$ as above
(see also Remark~\ref{GConvergence}):
\begin{corollary}
\label{MFAC}
Under the assumptions of Theorem 4.1 of~\cite{BMM07}, 
for any fixed $t$, almost surely,
$
\lim_{N\rightarrow\infty} \map\calD\aN(t) = \map\calD(t). 
\endstat
$
\end{corollary}

\begin{remark}
\label{APPR}
A consequence of
Corollary~\ref{MFAC} is that, under the assumptions of Theorem 4.1 of~\cite{BMM07}, 
for any fixed $t$, almost surely, 
for $N$ to $\infty$, we have that,
for all $\langle c,\vct{m}\rangle \in \at_{\map\calD\aN}(t,c)$ and $c' \in \sc{\os}$,
$
\sum_{\langle c',\vct{m}'\rangle: {\at_{\map\calD\aN}}(t+1,c')}
\hdtm\aN_{\langle c,\vct{m}\rangle,\langle c',\vct{m}' \rangle}
$
approaches
$\otm(\boldsymbol{\mu}(t))_{c,c'}$, i.e. for all $\epsilon>0$ there exists
$N_0$ s.t. for all $N \geq N_0$
$$
\left|
\left(
\sum_{\langle c',\vct{m}'\rangle: {\at_{\map\calD\aN}}(t+1,c')}
\hdtm\aN_{\langle c,\vct{m}\rangle,\langle c',\vct{m}'\rangle}
\right)
-
\otm(\boldsymbol{\mu}(t))_{c,c'}
\right|
< 
\epsilon
$$
$\endstat$
\end{remark}

\ifTR%
In the sequel we state the main theorem of the present paper. 
For {\em large} population sizes $N$, the probability of the set of paths of 
$\map\calD\aN(t)$ satisfying a formula $\varphi$, {\em approaches} 
the probability of the set of paths of $\map\calD(t)$ that satisfy $\varphi$. 
In order to guarantee that a formula $\Phi$ of form $\Prob{\bowtie}{p}{\varphi}$ holds of a state $s$ of $\map\calD(t)$ iff it holds of the corresponding state of $\map\calD\aN(t)$, it is sufficient that  $\Phi$ is {\em safe} with reference to  $\map\calD(t)$:
\else
In the sequel we state the main theorem of the present paper, that relies on the notion of {\em formulae safety}, with w.r.t.  $\map\calD(t)$:
\fi
a formula $\Phi$ is {\em safe} for a model $\calM$ iff for all sub-formulae $\Phi'$ of $\Phi$
and states $s$ of $\calM$, if $\Phi'$ is of the form $\Prob{\bowtie}{p}{\varphi}$ then  
$\Pr \{ \eta \in \Paths_{\calM}(s) \mid \eta \models_{\calM} \varphi \} \not= p$.

The theorem, together with Theorem~\ref{LMC:CORR} and  Lemma~\ref{CorrAbs}, establishes the formal relationship between the satisfaction relation on the 
exact semantics of the language and that on its mean-field approximation, thus justifying
the fast local model-checking instantiation we will show in the sequel.

\begin{theorem}
\label{CorrAprx}
Under the assumptions of 
Theorem 4.1 of~\cite{BMM07}, 
for all safe formulas $\Phi$,
for any fixed $t$ and  
$\map\aN(\vct{C}^{\aN}) \in \at_{\map\calD\aN}(t)$, 
almost surely,  for  $N$ large enough,
$
\map\aN(\vct{C}^{\aN}) \models_{\map\calD\aN} \Phi \mbox{ iff } 
\langle \vct{C}^{\aN}_{[1]}, \boldsymbol{\mu}(t) \rangle \models_{\map\calD} \Phi.
\endstat
$
\end{theorem}
\begin{proof}
The proof is carried out by induction on $\Phi$; in the proof we write $\vct{C}$ instead of $\vct{C}\aN$ for the sake of readability.\\[1em]
\ifTR
\noindent
{\bf Case $a$}:\\
The assert  follows directly from the definitions of $\left(\map\calD\aN(t)\right)_{N\geq N_0}$, 
$\models_{\map\calD\aN}$, and $\models_{\map\calD}$ (see also Remark~\ref{GConvergence}).\\[1em]
%
\noindent
{\bf Case $\neg \Phi$}:\\
The I. H. ensures that, for any fixed $t$ and $\map\aN(\vct{C}) \in \at_{\map\calD\aN}(t)$, 
a.s., there exists $\bar{N}$ s.t. for all $N \geq \bar{N}$, 
$\map\aN(\vct{C}) \models_{\map\calD\aN} \Phi$ iff
$\langle \vct{C}_{[1]}, \boldsymbol{\mu}(t) \rangle \models_{\map\calD} \Phi$.
But 
$\map\aN(\vct{C}) \models_{\map\calD\aN} \Phi$ iff
$\langle \vct{C}_{[1]}, \boldsymbol{\mu}(t) \rangle \models_{\map\calD} \Phi$
is logically equivalent to 
$$
\mbox{not }\map\aN(\vct{C}) \models_{\map\calD\aN} \Phi \mbox{ iff }
\mbox{not } \langle \vct{C}_{[1]}, \boldsymbol{\mu}(t) \rangle \models_{\map\calD} \Phi.
$$
Thus, by definition of $\models_{\map\calD\aN}$ and $\models_{\map\calD}$, we get
that, for any fixed $t$ and $\map\aN(\vct{C}) \in \at_{\map\calD\aN}(t)$, 
a.s., there exists $\bar{N}$ s.t. for all $N \geq \bar{N}$,
$\map\aN(\vct{C}) \models_{\map\calD\aN} \neg \, \Phi$ iff
$\langle \vct{C}_{[1]}, \boldsymbol{\mu}(t) \rangle \models_{\map\calD} \neg \, \Phi$.\\[1em]
\noindent
{\bf Case $\Phi_1 \vee \Phi_2$}:\\
The I. H. ensures that, for any fixed $t$ and $\map\aN(\vct{C}) \in \at_{\map\calD\aN}(t)$, 
a.s., there exists $\bar{N}_1$ s.t. for all $N \geq \bar{N}_1$, 
$\map\aN(\vct{C}) \models_{\map\calD\aN} \Phi_1$ iff
$\langle \vct{C}_{[1]}, \boldsymbol{\mu}(t) \rangle \models_{\map\calD} \Phi_1$, and
a.s., there exists $\bar{N}_2$ s.t. for all $N \geq \bar{N}_2$, 
$\map\aN(\vct{C}) \models_{\map\calD\aN} \Phi_2$ iff
$\langle \vct{C}_{[1]}, \boldsymbol{\mu}(t) \rangle \models_{\map\calD} \Phi_2$.\\
Let us now suppose  
$\map\aN(\vct{C}) \models_{\map\calD\aN} \Phi_1 \vee \Phi_2$ holds, i.e.
$\map\aN(\vct{C}) \models_{\map\calD\aN} \Phi_1$ holds or
$\map\aN(\vct{C}) \models_{\map\calD\aN} \Phi_2$ holds, by definition
of $\models_{\map\calD\aN}$. 

Suppose
$\map\aN(\vct{C}) \models_{\map\calD\aN} \Phi_1$ holds and, by the
I.H., we know that, a.s.  there exists $\bar{N}_1$ s.t. for all $N \geq \bar{N}_1$
$\langle \vct{C}_{[1]}, \boldsymbol{\mu}(t) \rangle \models_{\map\calD} \Phi_1$
holds as well. But then, we get that, for all such $N$, also
$\langle \vct{C}_{[1]}, \boldsymbol{\mu}(t) \rangle \models_{\map\calD} \Phi_1  \vee \Phi_2$ holds, by definition of $\models_{\map\calD}$. If, instead $\map\aN(\vct{C}) \models_{\map\calD\aN} \Phi_2$ holds,
we get the same result, using $\bar{N}_2$ instead of $\bar{N}_1$. Thus, for any
fixed $t$ and $\map\aN(\vct{C}) \in \at_{\map\calD\aN}(t)$, 
a.s. there exists $N \geq \max\SET{\bar{N}_1, \bar{N}_2}$ such that
if $\map\aN(\vct{C}) \models_{\map\calD\aN} \Phi_1 \vee \Phi_2$ holds, then
$\langle \vct{C}_{[1]}, \boldsymbol{\mu}(t) \rangle \models_{\map\calD} \Phi_1 \vee \Phi_2$ holds. \\
The proof for the reverse implication is similar.\\
\noindent
{\bf Case $\Prob{\bowtie}{p}{\lnxt \, \Phi}$}:\\
By definition of $\models_{\map\calD\aN}$ and $\models_{\map\calD}$,
we have to show that, for any fixed $t$ and $\map\aN(\vct{C}) \in \at_{\map\calD\aN}(t)$,
a.s., for $N$ large enough,
$$
\Pr \{ \rho \in \Paths_{\map\calD\aN}(\map\aN(\vct{C})) \mid \rho \models_{\map\calD\aN} \lnxt \, \Phi \} \bowtie p
$$ 
iff
$$
\Pr \{ \eta \in \Paths_{\map\calD}(\langle \vct{C_{[1]}},\boldsymbol{\mu}(t)\rangle) \mid \eta \models_{\map\calD} \lnxt \, \Phi \} \bowtie p.
$$
Below, we actually prove that, for any fixed $t$ and $\map\aN(\vct{C}) \in \at_{\map\calD\aN}(t)$,
 a.s., for $N$ large enough,
the probabilities of the two sets of paths are approaching each other, which implies the assert.\\
\noindent
$
\Pr \{ \rho \in \Paths_{\map\calD\aN}(\map\aN(\vct{C})) \mid \rho \models_{\map\calD\aN} \lnxt \, \Phi \}
$
is defined as
\begin{equation}
\label{pHDN}
p\aN_{\hdtm} = \sum_{\map\aN(\vct{C}'): \map\aN(\vct{C}') \models_{\map\calD\aN} \Phi}\hdtm\aN_{\map\aN(\vct{C}),\map\aN(\vct{C}')}
\end{equation}
and
$
\Pr \{ \eta \in \Paths_{\map\calD}(\langle \vct{C_{[1]}},\boldsymbol{\mu}(t)\rangle) \mid \eta \models_{\map\calD} \lnxt \, \Phi \}
$ 
is defined as 
\begin{equation}
\label{pHD}
p(t)_{\otm}= \sum_{\vct{C}'_{[1]}: \langle \vct{C}'_{[1]}, \boldsymbol{\mu}(t+1) \rangle \models_{\map\calD} \Phi}\otm(\boldsymbol{\mu}(t))_{\vct{C_{[1]}},\vct{C}'_{[1]}}.
\end{equation}
The I.H. ensures that, 
a.s., for $N \geq \bar{N}_{\vct{C}'}$, 
$\map\aN(\vct{C}') \models_{\map\calD\aN} \Phi$ if and only if
$\langle \vct{C}'_{[1]}, \boldsymbol{\mu}(t+1) \rangle \models_{\map\calD} \Phi$,
with $\map\aN(\vct{C}') \in \at_{\map\calD\aN}(t+1)$.
In particular, it holds that, for any specific value $\bar{c}$ of $\vct{C}'_{[1]}$ above
and $\map\aN(\vct{C}') \in \at_{\map\calD\aN}(t+1,\bar{c})$, 
$\map\aN(\vct{C}') \models_{\map\calD\aN} \Phi$ if and only if
$\langle \bar{c}, \boldsymbol{\mu}(t+1) \rangle \models_{\map\calD} \Phi$,
that is: either {\em all} elements of $ \at_{\map\calD\aN}(t+1,\bar{c})$
satisfy $\Phi$ or {\em none} of them does it.
Furthermore, for such $\bar{c}$, by Corollary~\ref{MFAC}, for all $\epsilon>0$ there exists
$N_{\bar{c}}$ s.t. for all $N \geq N_{\bar{c}}$  $$
\left|
\left(
\sum_{\langle \bar{c},\overline{\vct{m}}\rangle: {\at_{\map\calD\aN}}(t+1,\bar{c})}
\hdtm\aN_{\map\aN(\vct{C}),\langle \bar{c},\overline{\vct{m}}\rangle}
\right)
-
\otm(\boldsymbol{\mu}(t))_{\vct{C}_{[1]},\bar{c}}
\right|
< 
\epsilon
$$
(see Remark~\ref{}). So, for any $\epsilon > 0$ 
there exists an $\hat{N}$ larger than any of such
$\bar{N}_{\vct{C}'}$ and $N_{\bar{c}}$, such that for  all $N \geq \hat{N}$  
$
\left| p\aN_{\hdtm} - p(t)_{\otm} \right| < \epsilon
$
i.e.  the value $p\aN_{\hdtm}$ of sum~(\ref{pHDN}) approaches the value $p(t)_{\otm}$ of sum~(\ref{pHD}).
Finally, safety of $\Prob{\bowtie}{p}{\lnxt \, \Phi}$, implies that the value $p(t)_{\otm}$ of~(\ref{pHD}) is 
different from $p$. If $p(t)_{\otm} > p$ then we can choose $\epsilon$ small enough that also 
$p\aN_{\hdtm} > p$ and, similarly,  if $p(t)_{\otm} < p$, we get also $p\aN_{\hdtm} < p$, which proves the assert.\\
\noindent
{\bf Case $\Prob{\bowtie}{p}{\Phi_1 \rU{\le k}{} \Phi_2}$}:\\
\noindent
By definition of $\models_{\map\calD\aN}$ and $\models_{\map\calD}$,
we have to show that, for any fixed $t$ and $\map\aN(\vct{C}) \in \at_{\map\calD\aN}(t)$, 
a.s., for $N$ large enough,
$$
\Pr \{ \rho \in \Paths_{\map\calD\aN}(\map\aN(\vct{C})) \mid \rho \models_{\map\calD\aN} \Phi_1 \rU{\le k}{} \Phi_2 \} \bowtie p
$$ 
iff
$$
\Pr \{ \eta \in \Paths_{\map\calD}(\langle \vct{C_{[1]}},\boldsymbol{\mu}(t)\rangle) \mid \eta \models_{\map\calD} \Phi_1 \rU{\le k}{} \Phi_2 \} \bowtie p.
$$
Below, we actually prove that, for any fixed $t$ and $\map\aN(\vct{C}) \in \at_{\map\calD\aN}(t)$, 
a.s., for $N$ large enough,
the probabilities of the two sets of paths are approaching each other, which implies the assert.
We proceed by induction on $k$, using also the induction hypothesis on the structure of 
the formulas, when necessary.\\
\noindent
{\bf Base case ($k=0$)}:\\
$
\deriv
\Pr \{ \rho \in \Paths_{\map\calD\aN}(\map\aN(\vct{C})) \mid 
\rho \models_{\map\calD\aN} \, \Phi_1 \rU{\le 0}{} \Phi_2\}
\hint{=}{Def. of $\rho \models_{\map\calD\aN} \, \Phi_1 \rU{\le k}{} \Phi_2$}
\Pr \{ \rho \in \Paths_{\map\calD\aN}(\map\aN(\vct{C})) \mid 
\rho[0] \models_{\map\calD\aN}  \Phi_2\}
\hint{=}{Def. of $\Paths_{\map\calD\aN}(\map\aN(\vct{C}))$ and $\rho[0]$}
\left\{
\begin{array}{l}
1, \mbox{ if } \map\aN(\vct{C})  \models_{\map\calD\aN}  \Phi_2,\\\\
0, \mbox{ if } \mbox{ not }\map\aN(\vct{C})  \models_{\map\calD\aN}  \Phi_2
\end{array}
\right.
$\\
\noindent
By the I.H. on $\Phi_2$, 
with  $\map\aN(\vct{C}) \in \at_{\map\calD\aN}(t)$,
$\map\aN(\vct{C}) \models_{\map\calD\aN}  \Phi_2$ iff
$\langle \vct{C_{[1]}},\boldsymbol{\mu}(t)\rangle  \models_{\map\calD}  \Phi_2$, 
i.e., 
a.s.\\[1em]
\noindent
$
\deriv
\Pr \{ \rho \in \Paths_{\map\calD\aN}(\map\aN(\vct{C})) \mid 
\rho[0] \models_{\map\calD\aN}  \Phi_2\}
\hint{=}{See above}
\left\{
\begin{array}{l}
1, \mbox{ if } \langle \vct{C}_{[1]}, \boldsymbol{\mu}(t)\rangle \models_{\map\calD}  \Phi_2,\\\\
0, \mbox{ if } \mbox{ not }\langle \vct{C}_{[1]}, \boldsymbol{\mu}(t)\rangle   \models_{\map\calD}  \Phi_2
\end{array}
\right.
\hint{=}{Def. of $\Paths_{\map\calD}(\langle \vct{C}_{[1]}, \boldsymbol{\mu}(t)\rangle)$ and $\eta[0]$}
\Pr \{ \eta \in \Paths_{\map\calD}(\langle \vct{C}_{[1]}, \boldsymbol{\mu}(t)\rangle) \mid 
\eta[0] \models_{\map\calD}  \Phi_2\}
\hint{=}{Def. of $\eta \models_{\map\calD} \, \Phi_1 \rU{\le k}{} \Phi_2$}
\Pr \{ \eta \in \Paths_{\map\calD}(\langle \vct{C}_{[1]}, \boldsymbol{\mu}(t)\rangle) \mid 
\eta \models_{\map\calD} \, \Phi_1 \rU{\le 0}{} \Phi_2\}
$\\[1em]
\noindent
{\bf Induction step}:\\
\noindent
$
\deriv
\Pr \{ \rho \in \Paths_{\map\calD\aN}(\map\aN(\vct{C})) \mid 
\rho \models_{\map\calD\aN} \, \Phi_1 \rU{\le k+1}{} \Phi_2\}
\hint{=}{Def. $\rho \models_{\map\calD\aN} \, \Phi_1 \rU{\le k+1}{} \Phi_2$}
\left\{
\begin{array}{l}
0, \mbox{ if not } \map\aN(\vct{C})  \models_{\map\calD\aN}  \Phi_1 
\mbox{ and  not } \map\aN(\vct{C})  \models_{\map\calD\aN}  \Phi_2\\\\
1, \mbox{ if } \map\aN(\vct{C})  \models_{\map\calD\aN}  \Phi_2\\\\
\sum_{\map\aN(\vct{C}'): \map\aN(\vct{C}') \models_{\map\calD\aN} \Phi_1}\\
\mbox{\hspace{0.5in}}\hdtm\aN_{\map\aN(\vct{C}),\map\aN(\vct{C}')} \cdot
\Pr \{ \rho' \in \Paths_{\map\calD\aN}(\map\aN(\vct{C}')) \mid 
\rho' \models_{\map\calD\aN} \, \Phi_1 \rU{\le k}{} \Phi_2\},\\
\mbox{\hspace{0.5in}}\mbox{otherwise.}
\end{array}
\right.
$\\
\noindent
By the I.H. on $k$, 
noting that we are concerned only with those 
$\map\aN(\vct{C}')$ belonging to $\at_{\map\calD\aN}(t+1)$,
a.s., there is $\bar{N}$ s.t. for all
$N \geq \bar{N}$,\\[1em]
$
\Pr \{ \rho' \in \Paths_{\map\calD\aN}(\map\aN(\vct{C}')) \mid 
\rho' \models_{\map\calD\aN} \, \Phi_1 \rU{\le k}{} \Phi_2\}$ approaches\\$
\Pr \{ \eta' \in \Paths_{\map\calD}(\langle \vct{C}' _{[1]}, \boldsymbol{\mu}(t+1) \rangle) \mid 
\eta' \models_{\map\calD} \, \Phi_1 \rU{\le k}{} \Phi_2\}
$\\[1em]
Thus, \\[1em]
\noindent
$
\deriv
\Pr \{ \rho \in \Paths_{\map\calD\aN}(\map\aN(\vct{C})) \mid 
\rho \models_{\map\calD\aN} \, \Phi_1 \rU{\le k+1}{} \Phi_2\}
\hint{=}{See above}
\left\{
\begin{array}{l}
0, \mbox{ if not } \map\aN(\vct{C})  \models_{\map\calD\aN}  \Phi_1 
\mbox{ and  not } \map\aN(\vct{C})  \models_{\map\calD\aN}  \Phi_2\\\\
1, \mbox{ if } \map\aN(\vct{C})  \models_{\map\calD\aN}  \Phi_2\\\\
\sum_{\map\aN(\vct{C}'): \map\aN(\vct{C}') \models_{\map\calD\aN} \Phi_1}\\
\mbox{\hspace{0.5in}}\hdtm\aN_{\map\aN(\vct{C}),\map\aN(\vct{C}')} \cdot
\Pr \{ \eta' \in \Paths_{\map\calD}(\langle \vct{C}' _{[1]}, \boldsymbol{\mu}(t+1) \rangle) \mid 
\eta' \models_{\map\calD} \, \Phi_1 \rU{\le k}{} \Phi_2\},\\
\mbox{\hspace{0.5in}}\mbox{otherwise.}
\end{array}
\right.
$\\
The I.H. ensures that, 
a.s., there exist 
$N_1$, $N_2$ and a set of values $N_{\vct{C}'}$, for $\vct{C}'$ as in the sum above,
s.t. 
\begin{itemize}
\item
for all $N\geq N_1$,
$\map\aN(\vct{C}) \models_{\map\calD\aN} \Phi_1$ iff
$\langle \vct{C}_{[1]}, \boldsymbol{\mu}(t) \rangle \models_{\map\calD} \Phi_1$
\item
for all $N\geq N_2$,
$\map\aN(\vct{C}) \models_{\map\calD\aN} \Phi_2$ iff
$\langle \vct{C}_{[1]}, \boldsymbol{\mu}(t) \rangle \models_{\map\calD} \Phi_2$
\item
for all $N \geq N_{\vct{C}'}$,
$\map\aN(\vct{C}') \models_{\map\calD\aN} \Phi_1$ iff
$\langle \vct{C}'_{[1]}, \boldsymbol{\mu}(t+1) \rangle \models_{\map\calD} \Phi$.
\end{itemize}
%
%
Furthermore, by Corollary~\ref{MFAC}, using similar arguments as those
used for the case $\Prob{\bowtie}{p}{\lnxt \, \Phi}$, 
we get that a.s. there exists $\hat{N}$
such that, for $N \geq \hat{N}$,\\
$
\sum_{\map\aN(\vct{C}'): \map\aN(\vct{C}') \models_{\map\calD\aN} \Phi_1}\\
\mbox{\hspace{0.2in}}\hdtm\aN_{\map\aN(\vct{C}),\map\aN(\vct{C}')} \cdot
\Pr \{ \eta' \in \Paths_{\map\calD}(\langle \vct{C}' _{[1]}, \boldsymbol{\mu}(t+1) \rangle) \mid 
\eta' \models_{\map\calD} \, \Phi_1 \rU{\le k}{} \Phi_2\}
$\\\\
approaches\\\\
$
\sum_{\langle \vct{C}'_{[1]}, \boldsymbol{\mu}(t+1) \rangle \models_{\map\calD} \Phi_1}\\
\mbox{\hspace{0.2in}}\otm(\boldsymbol{\mu}(t))_{\vct{C_{[1]}},\vct{C}'_{[1]}} \cdot
\Pr \{ \eta' \in \Paths_{\map\calD}(\langle \vct{C}' _{[1]}, \boldsymbol{\mu}(t+1) \rangle) \mid 
\eta' \models_{\map\calD} \, \Phi_1 \rU{\le k}{} \Phi_2\}
$.
Thus, a.s.  for $N \geq \max\{N_1,N_2, \bar{N}_{\vct{C}'},\hat{N}\}$, with 
$\hat{N}, \bar{N}_{\vct{C}'} \geq N\vct{C}'$ for $\vct{C}'$ as above the following holds:
\begin{itemize}
\item 
$
\mbox{not } \map\aN(\vct{C})  \models_{\map\calD\aN}  \Phi_1 
\mbox{ and  not } \map\aN(\vct{C})  \models_{\map\calD\aN}  \Phi_2
$ iff\\
$
\mbox{not } \langle \vct{C}_{[1]},\boldsymbol{\mu}(t) \rangle  \models_{\map\calD}  \Phi_1 \mbox{ and  not }
\langle \vct{C}_{[1]},\boldsymbol{\mu}(t) \rangle  \models_{\map\calD}  \Phi_2
$
\item
$
\map\aN(\vct{C})  \models_{\map\calD\aN}  \Phi_2
$ iff
$
\langle \vct{C}_{[1]},\boldsymbol{\mu}(t) \rangle  \models_{\map\calD}  \Phi_2
$
\item
$
\sum_{\map\aN(\vct{C}'): \map\aN(\vct{C}') \models_{\map\calD\aN} \Phi_1}\\
\mbox{\hspace{0.2in}}\hdtm\aN_{\map\aN(\vct{C}),\map\aN(\vct{C}')} \cdot
\Pr \{ \eta' \in \Paths_{\map\calD}(\langle \vct{C}' _{[1]}, \boldsymbol{\mu}(t+1) \rangle) \mid 
\eta' \models_{\map\calD} \, \Phi_1 \rU{\le k}{} \Phi_2\}
$\\\\
approaches\\\\
$
\sum_{\langle \vct{C}'_{[1]}, \boldsymbol{\mu}(t+1) \rangle \models_{\map\calD} \Phi_1}\\
\mbox{\hspace{0.2in}}\otm(\boldsymbol{\mu}(t))_{\vct{C_{[1]}},\vct{C}'_{[1]}} \cdot
\Pr \{ \eta' \in \Paths_{\map\calD}(\langle \vct{C}' _{[1]}, \boldsymbol{\mu}(t+1) \rangle) \mid 
\eta' \models_{\map\calD} \, \Phi_1 \rU{\le k}{} \Phi_2\}
$
\end{itemize}
and by safety of $\Prob{\bowtie}{p}{\Phi_1 \rU{\le k+1}{} \Phi_2}$ we get the assert.
\hfill$\Box$
\else
For brevity, we show only the case for $\Prob{\bowtie}{p}{\lnxt \, \Phi}$; for the complete proof we refer to~\cite{La+13a}.  
By definition of $\models_{\map\calD\aN}$ and $\models_{\map\calD}$,
we have to show that, for any fixed $t$ and $\map\aN(\vct{C}) \in \at_{\map\calD\aN}(t)$,
a.s., for $N$ large enough,
$$
\Pr \{ \rho \in \Paths_{\map\calD\aN}(\map\aN(\vct{C})) \mid \rho \models_{\map\calD\aN} \lnxt \, \Phi \} \bowtie p
$$ 
iff
$$
\Pr \{ \eta \in \Paths_{\map\calD}(\langle \vct{C_{[1]}},\boldsymbol{\mu}(t)\rangle) \mid \eta \models_{\map\calD} \lnxt \, \Phi \} \bowtie p.
$$
Below, we actually prove that, for any fixed $t$ and $\map\aN(\vct{C}) \in \at_{\map\calD\aN}(t)$,
 a.s., for $N$ large enough,
the probabilities of the two sets of paths are approaching each other, which implies the assert.\\
\noindent
$
\Pr \{ \rho \in \Paths_{\map\calD\aN}(\map\aN(\vct{C})) \mid \rho \models_{\map\calD\aN} \lnxt \, \Phi \}
$
is defined as
\begin{equation}
\label{pHDN}
p\aN_{\hdtm} = \sum_{\map\aN(\vct{C}'): \map\aN(\vct{C}') \models_{\map\calD\aN} \Phi}\hdtm\aN_{\map\aN(\vct{C}),\map\aN(\vct{C}')}
\end{equation}
and
$
\Pr \{ \eta \in \Paths_{\map\calD}(\langle \vct{C_{[1]}},\boldsymbol{\mu}(t)\rangle) \mid \eta \models_{\map\calD} \lnxt \, \Phi \}
$ 
is defined as 
\begin{equation}
\label{pHD}
p(t)_{\otm}= \sum_{\vct{C}'_{[1]}: \langle \vct{C}'_{[1]}, \boldsymbol{\mu}(t+1) \rangle \models_{\map\calD} \Phi}\otm(\boldsymbol{\mu}(t))_{\vct{C_{[1]}},\vct{C}'_{[1]}}.
\end{equation}
The I.H. ensures that, 
a.s., for $N \geq \bar{N}_{\vct{C}'}$, 
$\map\aN(\vct{C}') \models_{\map\calD\aN} \Phi$ if and only if
$\langle \vct{C}'_{[1]}, \boldsymbol{\mu}(t+1) \rangle \models_{\map\calD} \Phi$,
with $\map\aN(\vct{C}') \in \at_{\map\calD\aN}(t+1)$.
In particular, it holds that, for any specific value $\bar{c}$ of $\vct{C}'_{[1]}$ above
and $\map\aN(\vct{C}') \in \at_{\map\calD\aN}(t+1,\bar{c})$, 
$\map\aN(\vct{C}') \models_{\map\calD\aN} \Phi$ if and only if
$\langle \bar{c}, \boldsymbol{\mu}(t+1) \rangle \models_{\map\calD} \Phi$,
that is: either {\em all} elements of $ \at_{\map\calD\aN}(t+1,\bar{c})$
satisfy $\Phi$ or {\em none} of them does it.
Furthermore, for such $\bar{c}$, by Corollary~\ref{MFAC}, for all $\epsilon_{\bar{c}}>0$ there exists
$N_{\bar{c}}$ s.t. for all $N \geq N_{\bar{c}}$  $$
\left|
\left(
\sum_{\langle \bar{c},\overline{\vct{m}}\rangle: {\at_{\map\calD\aN}}(t+1,\bar{c})}
\hdtm\aN_{\map\aN(\vct{C}),\langle \bar{c},\overline{\vct{m}}\rangle}
\right)
-
\otm(\boldsymbol{\mu}(t))_{\vct{C}_{[1]},\bar{c}}
\right|
< 
\epsilon_{\bar{c}}
$$
(see Remark~\ref{APPR}). So, for any $\epsilon > 0$ 
there exists an $\hat{N}$ larger than any of such
$\bar{N}_{\vct{C}'}$ and $N_{\bar{c}}$, such that for  all $N \geq \hat{N}$  
$
\left| p\aN_{\hdtm} - p(t)_{\otm} \right| < \epsilon
$
i.e.  the value $p\aN_{\hdtm}$ of sum~(\ref{pHDN}) approaches the value $p(t)_{\otm}$ of sum~(\ref{pHD}).
Finally, safety of $\Prob{\bowtie}{p}{\lnxt \, \Phi}$, implies that the value $p(t)_{\otm}$ of~(\ref{pHD}) is 
different from $p$. If $p(t)_{\otm} > p$ then we can choose $\epsilon$ small enough that also 
$p\aN_{\hdtm} > p$ and, similarly,  if $p(t)_{\otm} < p$, we get also $p\aN_{\hdtm} < p$, which proves the assert.\\
\hfill$\Box$
\fi
\end{proof}
Finally, using Lemma~\ref{CorrAbs} we get the following
\begin{corollary}
Under the assumptions of Theorem 4.1 of~\cite{BMM07}, for all safe formulas $\Phi$,
for any fixed $t$ and  
$\vct{C}^{\aN} \in \at_{\calD\aN}(t)$, 
almost surely,  for  $N$ large enough
$
\vct{C}^{\aN} \models_{\calD\aN} \Phi \mbox{ iff } 
\langle \vct{C}^{\aN}_{[1]}, \boldsymbol{\mu}(t) \rangle \models_{\map\calD} \Phi.
\endstat
$
\end{corollary}

\paragraph{Fast  local model-checking} 
On-the-fly {\em fast} PCTL model-checking on the limit DTMC $\map\calD(t)$ is obtained by instantiating \textsf{proc} with  
$\sc{\os} \times \us{S}
$
and \textsf{lab} with 
$\sls_1 \cup \sls_g$;
$\nxt$ is instantiated with $\nxt_{\map\calD}$ defined as follows:
$$
\nxt_{\map\calD} (\langle c,\vct{m} \rangle)=
[(\langle c',\vct{m}\cdot\otm(\vct{m}) \rangle, p') \mid \otm(\vct{m})_{c,c'} = p' > 0],
$$
\noindent
with $\otm(\vct{m})_{c,c'}$ as in Theorem 4.1 of~\cite{BMM07}; 
$\lab$ is instantiated 
as expected:
$
\lab_{\map\calD}(\langle c,\vct{m} \rangle,a)= a \in \slf_{\map\calD}(\langle c,\vct{m} \rangle)
$. 
The instantiation is implemented in \FlyFast{.}
\begin{remark}
Although in the hypothesis of the theorem we require formulae safety, 
for all practical purposes, it is actually sufficient
to require that 
$$
\Pr \{ \eta \in \Paths_{\map\calD}(s') \mid \eta \models_{\map\calD} \varphi \} \not= p
$$
for all formulae $\Prob{\bowtie}{p}{\varphi}$ and states $s'$ such that 
$\checkPath(s',\varphi)$ is computed during the execution of $\checkDtmc(s,\Phi)$(see Table~\ref{alg:check_dtmc}).
This (weaker) safety check is readily added to the algorithm.
$\endstat$
\end{remark}
\begin{example}[\FlyFast{} results]
\label{ex:fmcatwork} 
Fig.~\ref{tab:fastmcatwork} shows the result of \FlyFast{} on the model of Ex.~\ref{ex:epidemicmodel} for the first object of a large population of objects, each initially in state $S$.
In Fig.~\ref{tab:fastmcatwork} (left) the same properties are considered as in Ex.~\ref{ex:properties}.
The analysis takes less than a second and is {\em insensitive} to the total population size. Fig.~\ref{tab:fastmcatwork} (right) shows how the probability measure of the set of paths satisfying formula 
$~true~\rU{\le k} ( !e \wedge !i\wedge \Prob{>}{0.3}{~true~\rU{\le 5}~i~})$  of property
 \emph{P3} on page~\pageref{P3}, (for $k=3$), changes for initial time $t0$ varying from $0$ to $10$.\\

\end{example}
\vspace{-2em}
\begin{figure}[tbp]
\begin{tabular}{cc}
\begin{minipage}{5.5cm}
\includegraphics[scale=0.25]{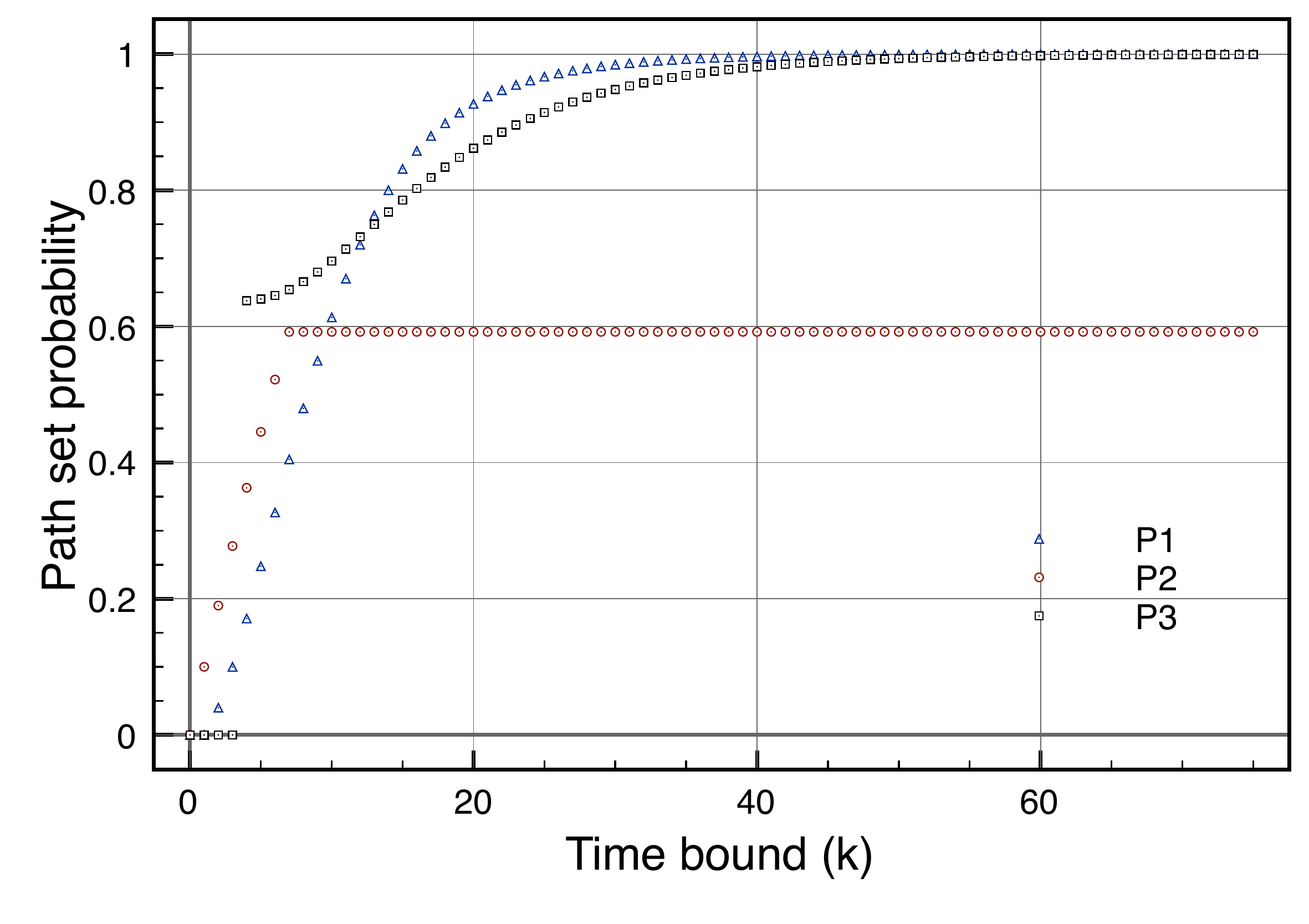} 
\end{minipage} &
\begin{minipage}{5.5cm}
\includegraphics[scale=0.25]{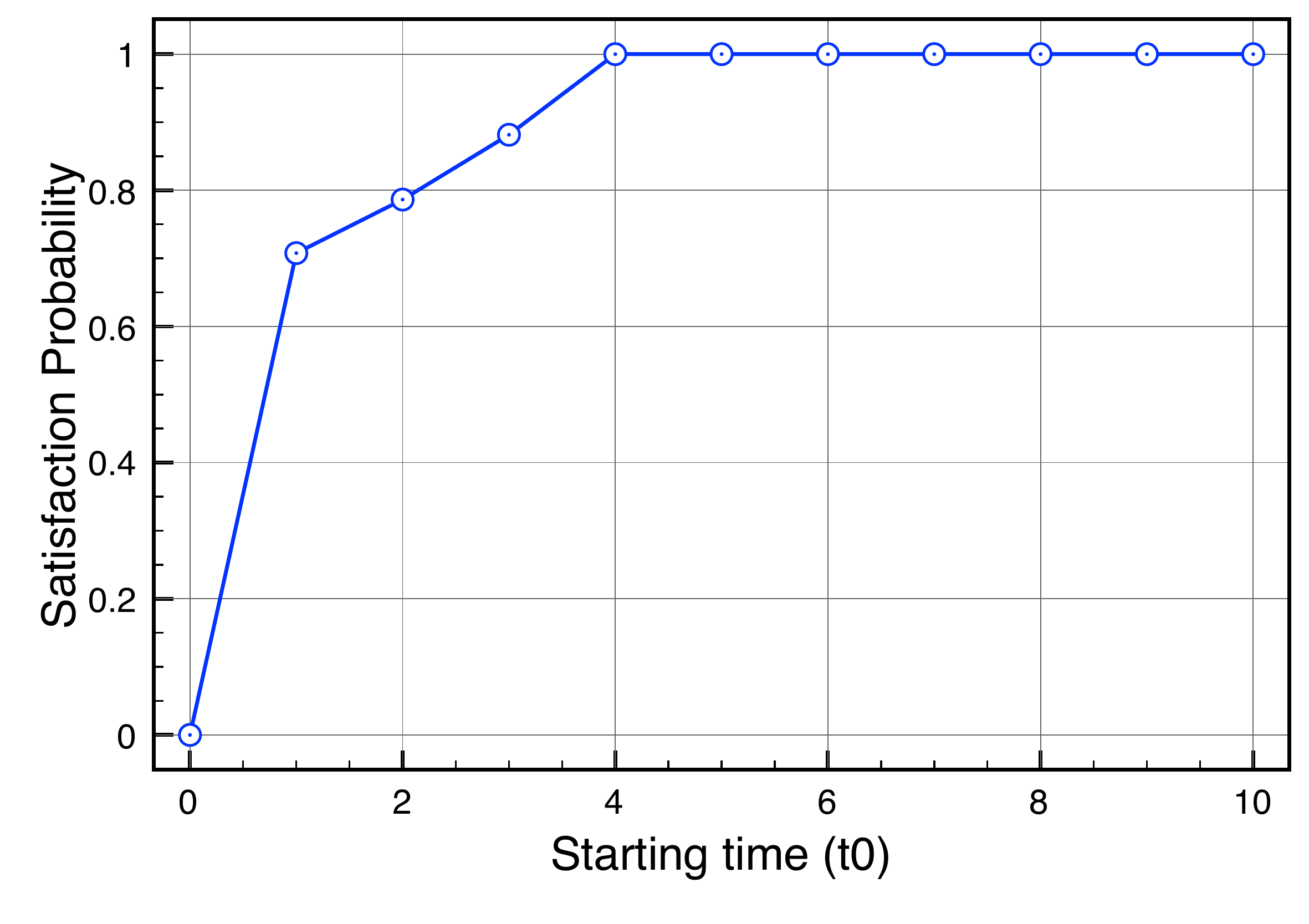} 
\end{minipage} 
\end{tabular}
\caption{Fast model-checking results.}
\label{tab:fastmcatwork}
\end{figure}

\section{Conclusions and Future Work}
\label{CONC:FW}

In this paper we have presented a fast PCTL model-checking approach that builds upon
{\em local, on-the-fly} model-checking and {\em mean-field} approximation, allowing
for scalable analysis of selected objects in the context of very large systems. The model-checking algorithm is parametric w.r.t. the specific semantic model of interest. We presented related correctness results, an example of application of a prototype implementation and briefly discussed complexity of the algorithm. The results can be trivially extended in order to consider multiple
selected objects. Following approaches similar to those presented in~\cite{BMM07},
we plan to extend our work to heterogeneous systems and systems with memory.
We are interested in extensions that address spatial distribution of objects as well as more expressive logics, combining local and global properties, and languages (e.g.~\cite{MNS11,Ko+13}) and to study the exact relation between mean field convergence results for continuous interleaving models and discrete, time-synchronous ones.

\bibliographystyle{splncs03}
\bibliography{abbr,opmc,xref}

\end{document}